\documentclass[12pt]{article}
\usepackage{amsmath}
\usepackage{graphicx,psfrag,epsf}
\usepackage{enumerate}
\usepackage{url} % not crucial - just used below for the URL 

\newcommand{\blind}{0}

\usepackage{amsthm,amsfonts,amssymb}
\usepackage{comment}
\usepackage[bf]{caption}
\usepackage{lscape}
\usepackage{color}
\usepackage[backend=bibtex, style=nature, citestyle=authoryear-comp]{biblatex}
\bibliography{bibliography}
\allowdisplaybreaks

\def\indist{\rightsquigarrow}
\def\ind{\perp\!\!\!\perp}

\newcommand{\Pb}{\mathbb{P}}

\newcommand{\Pn}{\mathbb{P}_n}
\newcommand{\Gn}{\mathbb{G}_n}
\newcommand{\Gb}{\mathbb{G}}

\newcommand{\E}{\mathbb{E}}
\newcommand{\R}{\mathbb{R}}

\newcommand{\bZ}{\mathbf{Z}}
\newcommand{\bz}{\mathbf{z}}
\newcommand{\bX}{\mathbf{X}}
\newcommand{\bx}{\mathbf{x}}
\newcommand{\bH}{\mathbf{H}}
\newcommand{\bh}{\mathbf{h}}

\newcommand{\bQ}{\mathbf{Q}}
\newcommand{\bD}{\mathbf{D}}

\def\expit{\text{expit}}

\DeclareMathOperator*{\argmax}{arg\,max}

\DeclareSymbolFont{bbold}{U}{bbold}{m}{n}
\DeclareSymbolFontAlphabet{\mathbbold}{bbold}
\newcommand{\one}{\mathbbold{1}}
\newtheorem{theorem}{Theorem}
\newtheorem{lemma}{Lemma}
\newtheorem{corollary}{Corollary}

\newtheorem{algorithm}{Algorithm}
\theoremstyle{remark}
\newtheorem{assumption}{Assumption}

\newtheorem{remark}{Remark}

\begin{document}

\def\spacingset#1{\renewcommand{\baselinestretch}%
{#1}\small\normalsize} \spacingset{1}

%%%%%%%%%%%%%%%%%%%%%%%%%%%%%%%%%%%%%%%%%%%

\if0\blind
{
  \title{ \bf Nonparametric causal effects based on incremental propensity score interventions}
    \author{Edward H. Kennedy
    \thanks{Edward Kennedy is Assistant Professor in the Department of Statistics, Carnegie Mellon University, Pittsburgh, PA 15213 (e-mail: edward@stat.cmu.edu). The author thanks Traci Kennedy, Miguel Hernan, Kwangho Kim, and the Causal Inference Reading Group at Carnegie Mellon 
    for helpful discussions and comments, and Valerio Bacak for guidance on the National Longitudinal Survey of Youth data analysis.}\hspace{.2cm}
    \\
    Department of Statistics, Carnegie Mellon University \\ \\
    }
  \maketitle
  \setcounter{page}{0}
  \thispagestyle{empty}
} \fi

\if1\blind
{
  \vspace*{.8in}
  \begin{center}
    {\LARGE\bf Nonparametric causal effects based on \\ \vspace{.15in} incremental propensity score interventions  }
\end{center}
  \setcounter{page}{0}
  \medskip
} \fi

\vspace{-.2in}
\begin{abstract}
Most work in causal inference considers deterministic interventions that set each unit's treatment to some fixed value. However, under positivity violations these interventions can lead to non-identification, inefficiency, and effects with little practical relevance. Further, corresponding effects in longitudinal studies are highly sensitive to the curse of dimensionality, resulting in widespread use of unrealistic parametric models. We propose a novel solution to these problems: incremental interventions that shift propensity score values rather than set treatments to fixed values. Incremental interventions have several crucial advantages. First, they avoid positivity assumptions entirely. Second, they require no parametric assumptions and yet still admit a simple characterization of longitudinal effects, independent of the number of timepoints. For example, they allow longitudinal effects to be visualized with a single curve instead of lists of coefficients. After characterizing incremental interventions and giving identifying conditions for corresponding effects, we also develop general efficiency theory, propose efficient nonparametric estimators that can attain fast convergence rates even when incorporating flexible machine learning, and propose a bootstrap-based confidence band and simultaneous test of no treatment effect. Finally we explore finite-sample performance via simulation, and apply the methods to study time-varying sociological effects of incarceration on entry into marriage.
\end{abstract}

\noindent%
\vfill

\thispagestyle{empty}

\newpage

%%% Most manuscripts accepted for publication are no more than 30 double-spaced pages!
%\spacingset{1.45} % DON'T change the spacing!
\spacingset{1}

\section{Introduction} 
\label{sec:intro}

Most work in causal inference considers deterministic interventions that set each unit's treatment to some fixed value. For example, the usual average treatment effect indicates how mean outcomes would change if all units were uniformly assigned treatment versus control. Similarly, standard marginal structural models \autocite{robins2000marginal} describe outcomes had all units followed given exposure trajectories over time (e.g., treated at every time, treated after time $t$, etc.). However, these simple effects are not identified when some units have zero chance to receive given treatment options. This is a violation of the so-called positivity assumption, which has been known in the causal inference literature since at least \textcite{rosenbaum1983central}. Even if positivity is only nearly violated (i.e., chances of some treatment options are merely small), the finite-sample behavior of many common estimators can be severely degraded \autocite{kang2007demystifying, moore2012causal}. 

Similarly, even if positivity holds, in longitudinal studies these standard effects are afflicted by a curse of dimensionality in the number of study timepoints: exponentially many samples are needed to learn about all treatment trajectories. For example, in a simple trial with a binary randomized treatment and ten timepoints, we would need %more than 5000 patients to have \textit{any} chance at seeing at least one in all $2^{10}$ exposure trajectories; and 
nearly 12,000 patients to guarantee $<$1\% chance of having any unrepresented exposure trajectories. The usual way to deal with this problem is to assume it away with a parametric model for how outcomes change across trajectories. However, such models are typically severely wrong if overly simple, and can be hard to interpret otherwise (and are often still misspecified). Further, in the real world, treatment is typically not applied uniformly, so static deterministic interventions may not be of practical policy interest. For example, most medical treatments would never be applied indiscriminately, but instead would be recommended or not based on characteristics of the patient and physician prescribing preference. 

Thus there has been substantial recent interest in dynamic and stochastic interventions, which can depend on unit characteristics and be random rather than deterministic. Examples have been studied for point exposures by \textcite{tian2008identifying, pearl2009causality, dudik2014doubly}, and in longitudinal studies by \textcite{murphy2001marginal, robins2004effects, van2007causal, robins2008estimation, taubman2009intervening, cain2010start, young2011comparative}. 
%\textcite{tian2008identifying, pearl2009causality, diaz2012population, moore2012causal, diaz2013assessing, haneuse2013estimation, dudik2014doubly}, and in longitudinal studies by \textcite{murphy2001marginal, robins2004effects, van2007causal, taubman2009intervening, cain2010start, young2011comparative, young2014identification}. 
Particularly relevant to this paper is work by \textcite{diaz2012population, moore2012causal, diaz2013assessing, haneuse2013estimation, young2014identification}, who consider interventions that depend on the observational treatment process. However, to the best of our knowledge, none of the existing intervention effects both avoids positivity conditions entirely and is completely nonparametric (even in studies with many timepoints). 

In this paper we propose novel \textit{incremental intervention effects}, based on shifting propensity scores rather than setting treatment values. We show that such effects can be identified and estimated without any positivity or parametric assumptions, and \textcolor{black}{argue that they can be more realistic than other interventions. One trade-off is that they yield effects that are more descriptive than prescriptive.} We develop nonparametric influence-function-based estimators that can incorporate flexible machine learning tools, while still providing valid parametric-rate inference. Our methods for uniform inference (using the multiplier bootstrap) also yield a new general test of no treatment effect. We conclude with a simulation study, and apply the methods in a longitudinal study of incarceration effects.

\setcounter{section}{1}
\section{Notation \& Setup}
\label{sec:setup}

We consider the case where we observe a sample $(\bZ_1,...,\bZ_n)$ of iid observations from distribution $\Pb$, with 
$$ \bZ = (\bX_1, A_1, \bX_2, A_2, ..., \bX_T, A_T, Y) $$
for covariates $\bX_t$ and treatment $A_t$ at time $t$, and outcome $Y$. For simplicity, at present we consider binary treatments (so the support of $A_t$ is $\mathcal{A}=\{0,1\}$) and completely observed $\bZ$ (so there is no missingness or dropout), but extensions will appear in future work. We use overbars to denote the past history of a variable, so that $\overline\bX_t=(\bX_1,...,\bX_t)$ and $\overline{A}_t=(A_1,...,A_t)$, for example, and we let $\bH_t=(\overline\bX_t, \overline{A}_{t-1})$ denote the past history just prior to treatment at time $t$, with support $\mathcal{H}_t$.

\begin{remark}
The above data setup also covers the case where outcomes are time-varying, i.e., where rather than $\bZ$ the observations are given by
$ \bZ^* = (\bX^*_1, A_1, Y_1,  ..., \bX^*_{T^*}, A_{T^*}, Y_{T^*})$. 
This follows since we can let $\bX_t = (\bX^*_t,Y_{t-1})$ and $Y=Y_{T^*}$ in our original formulation. If interest centers on treatment effects on an earlier outcome $Y_t$ (for $t<T^*$), rather than $Y_{T^*}$, then we can let $Y=Y_t$ and truncate the sequence, defining $T=t$ instead of $T=T^*$. 
\end{remark}

In this paper we use  potential outcomes \autocite{rubin1974estimating}, and so let $Y^{\overline{a}_T}$ denote the outcome that {would have} been observed had the treatment sequence $\overline{A}_T=\overline{a}_T$ been received. The quantity $Y^{\overline{a}_T}$ is an example of a counterfactual based on a \textit{deterministic static intervention}, in which a fixed treatment is applied with probability one and regardless of covariate information (e.g., $\overline{A}_T=\overline{a}_T$ is applied uniformly across units, regardless of past histories $\bH_t$). Deterministic static interventions are the kinds of interventions most commonly considered in practice; examples include the average effect of a point exposure $\E(Y^1-Y^0)$, and standard marginal structural model and structural nested model parameters $\E(Y^{\overline{a}_T})$ and $\E(Y^{\overline{a}_t,0} - Y^{\overline{a}_{t-1},0} \mid \bH_t,A_t)$, respectively \autocite{robins2000marginal, robins2000marginal2}. 

Alternatively, in \textit{deterministic dynamic interventions} \autocite{robins1986new, murphy2001marginal}  treatment at time $t$ is assigned according to a fixed rule $d_t: \mathcal{H}_t \mapsto \mathcal{A}$ that depends on past history. Characterizing and estimating the optimal such rule is a major goal in the optimal dynamic treatment regime literature \autocite{murphy2003optimal, robins2004optimal}. The potential outcome under a sequence of hypothetical rules $\textbf{d}=\overline{d}_T=(d_1,...,d_T)$ can be expressed as $Y^{\textbf{d}}$, where the dependence of the rules on the histories $\bH_t$ is suppressed for notational simplicity, and $\textbf{d}$ is lower-case since the rule is non-random (given the histories). A simple example is the rule $d_t = \one(V_t \geq c_t)$  that assigns treatment if a variable $V_t \subset \bX_t$ passes some threshold $c_t \in \R$, with corresponding mean outcome $\E(Y^{(d_1,...,d_T)})$.

In this paper we propose a new form of \textit{stochastic dynamic intervention}, which is an intervention where treatment at each time is randomly assigned based on a conditional distribution $q_t(a_t \mid \bh_t)$. Stochastic interventions can thus be viewed as random choices among deterministic rules. These interventions have not been studied as extensively as other types, with important exceptions listed in the Introduction (see for example \textcite{diaz2012population, haneuse2013estimation, young2014identification} for review). We express the potential outcome under a stochastic intervention as $Y^{\textbf{Q}}$, where $\textbf{Q}=(Q_1,...,Q_T)$ represents draws from the conditional distributions $q_t$, and is upper-case since the intervention is stochastic. A simple stochastic intervention related to the previous rule $d_t=\one(V_t \geq c_t)$ would be $Q_t = \one(V_t \geq C_t)$ where $C_t \sim N(0,1)$ is now a random threshold.

\section{Incremental Intervention Theory}
\label{sec:ipsi}

In this section we first describe a new class of stochastic dynamic intervention, which we call incremental propensity score interventions, and give some motivation and examples. We go on to show that these interventions are nonparametrically identified without requiring any positivity restrictions on the propensity scores (e.g., the propensity scores do not need to be bounded away from zero and one). Then we describe the efficiency theory for estimating mean outcomes under these interventions, based on a new result for general stochastic interventions that depend on the observational treatment  distribution.

\subsection{Proposed Interventions}

In this paper we propose \textit{incremental propensity score interventions} that replace the observational treatment process (i.e., propensity score) $\pi_t(\bh_t) = \Pb(A_t= 1 \mid \bH_t=\bh_t)$ with a shifted version, based on multiplying the odds of receiving treatment. Specifically, our proposed intervention replaces the observational propensity score $\pi_t$ with the distribution defined by
\begin{equation}
q_t( \bh_t; \delta, \pi_t) = \frac{ \delta \pi_t( \bh_t) }{ \delta \pi_t(\bh_t) + 1 - \pi_t(\bh_t) } \ \text{ for } 0 < \delta < \infty .
\label{eq:qdef}
\end{equation}
The increment parameter $\delta \in (0,\infty)$ is user-specified, and dictates the extent to which the propensity scores are fluctuated from their actual observational values. In practice we recommend considering a range of $\delta$ values, depending on the scientific question at hand. Although the intervention distribution $q_t$ depends on both the increment $\delta$ and the observational propensity $\pi_t$, we often drop this dependence and write $q_t(\bh_t; \delta,\pi_t) = q_t(\bh_t)$ or just $q_t$ to ease notation. As mentioned earlier, \textcite{diaz2012population, moore2012causal, diaz2013assessing, haneuse2013estimation, young2014identification} {have considered other different interventions that depend on the observational treatment process}.

\textcolor{black}{Our choice of $q_t$ in \eqref{eq:qdef} is motivated by both interpretability and mathematical convenience}. In particular it yields
$$ \delta = \frac{ q_t( \bh_t) / \{1 - q_t( \bh_t) \} } {\pi_t( \bh_t) / \{1- \pi_t( \bh_t)\}} = \frac{\text{odds}_q(A_t=1 \mid \bH_t=\bh_t)}{\text{odds}_\pi(A_t=1 \mid \bH_t=\bh_t)} $$
whenever $0 < \pi_t < 1$, % (otherwise $q_t=\pi_t$ if $\pi_t \in \{0,1\}$).
so the increment $\delta$ is just an odds ratio, indicating how the intervention changes the odds of receiving treatment. As with usual odds ratios, if $\delta> 1$ then the intervention increases the odds of receiving treatment, and if $\delta  < 1$ then the intervention decreases these odds (if $\delta=1$ then $q_t=\pi_t$ so the treatment process is left unchanged). For example, an intervention with $\delta =1.5$ would increase the odds of receiving treatment by 50\% for each patient with $0 < \pi_t < 1$: a patient with an actual 25\% chance of receiving treatment (1/3 odds) would instead have a  33\% chance (1/2 odds) under the intervention.

\textcolor{black}{In addition to other kinds of shifts (e.g., risk ratios) in future work we will consider interventions for which $\delta=\delta_t(\bh_t)$ can depend on time and covariate history. However, even with $\delta$ fixed, incremental interventions are still dynamic, since the conditional distribution $q_t$ depends on the covariate history. In other words,} these interventions are personalized to patient characteristics through the propensity score. For example, under an intervention with $\delta=1.5$, a patient with a 50\% chance of receiving treatment observationally would instead have a 60\% chance under the intervention, while a patient whose chances were 5\%  would only see an increase to 7.3\% \textcolor{black}{(i.e., multiplying the odds by a fixed factor yields different shifts in the probabilities). Contrast this with a usual static intervention, which flatly assigns all patients a particular sequence $\overline{a}_T$ (or a random choice among such sequences) regardless of propensity score.} Figure~\ref{fig:explot} illustrates  incremental interventions with data on $n=20$ simulated observations in a hypothetical study with $T=2$ timepoints. 

%\textcolor{black}{In addition to other kinds of shifts (e.g., based on risk ratios)}, in future work we will consider interventions for which $\delta=\delta_t(\bh_t)$ can depend on time and covariate history. However, even when $\delta$ is fixed and does not depend on such information, the corresponding intervention is still personalized to patient characteristics through the propensity score: although the factor multiplying the odds is the same for each history, the probabilities shift by different amounts. \textcolor{black}{Therefore incremental interventions are still dynamic even with fixed $\delta$, since the conditional distribution $q_t$ depends on the covariate history.}  For example, under an intervention with $\delta=1.5$, a patient with a 50\% chance of receiving treatment observationally would instead have a 60\% chance under the intervention, while a patient whose chances were 5\%  would only see an increase to 7.3\%. \textcolor{black}{Contrast this with a usual static intervention, which flatly assigns all patients a particular sequence $\overline{a}_T$ (or a random choice among such sequences) regardless of propensity score.} Figure~\ref{fig:explot} illustrates  incremental interventions with data on $n=20$ simulated observations in a hypothetical study with $T=2$ timepoints. 

\begin{figure}[h!]
\begin{center}
\includegraphics[width=\textwidth]{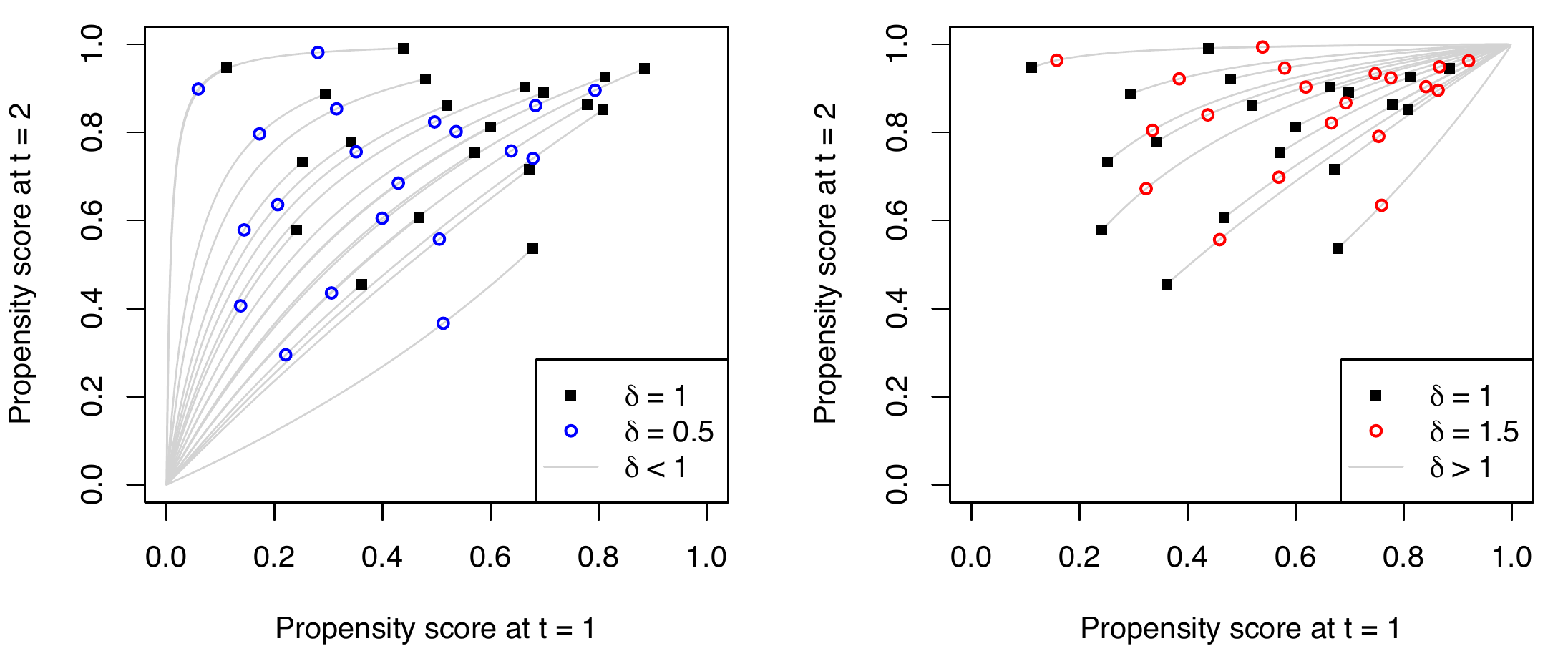}
\caption{Observational propensity scores for $n=20$ simulated units in a study with $T=2$ timepoints, and their values under incremental interventions based on different $\delta$ values ($\delta \leq 1$ in the left plot, $\delta \geq 1$ in the right).} \label{fig:explot}
\end{center}
\end{figure}

\vspace{-.05in}

Figure~\ref{fig:explot} also helps illustrate why incremental interventions require weak identifying assumptions (to be discussed shortly), and \textcolor{black}{can be more likely to occur in practice, compared to other kinds of interventions}. Although usual static interventions (e.g., setting $A=1$) might require forcing treatment on someone with only a 1\% chance of receiving it in the real world, the proposed incremental intervention only requires that the propensity score be slightly shifted (e.g., from 1\% to 1.5\% when $\delta=1.5$). \textcolor{black}{In settings where treatment changes occur more gradually (e.g., when physicians slightly reduce treatment intensity, or judges become slightly more lenient), incremental interventions might be similar to treatment changes that could occur naturally in practice. Even if not especially likely to occur, incremental interventions still might be more realistic than standard static interventions since they are ``closer'' to the observational treatment distribution. Of course, incremental interventions can still be useful analysis tools even if not necessarily mimicking realistic treatment changes, as discussed in more detail in the next subsection. }

\textcolor{black}{One trade-off between incremental and more standard interventions is that, by virtue of their dependence on the observational treatment process, incremental interventions will often play a more descriptive rather than prescriptive role. Incremental interventions  allow one to describe how outcomes would vary with gradual changes in treatment intensity; but they are typically less useful for making specific recommendations about optimal treatment.}

Nonetheless incremental interventions do generalize common static and dynamic interventions (both deterministic and stochastic), since they can recover these interventions with particular choices of $\delta_t(\bh_t)$. For example, if positivity holds  then taking the values $\delta=\infty$ and  $\delta=0$ recovers the usual static interventions, yielding potential outcomes $Y^\mathbf{1}$ and $Y^\mathbf{0}$ under exposures  $\overline{A}_T=(1,...,1)$ and $\overline{A}_T=(0,...,0)$, respectively. Thus incremental interventions can also be used for a sensitivity analysis of the positivity assumption. If positivity is violated, then $q_t \rightarrow \one(\pi_t>0)$ for $\delta \rightarrow \infty$, and $q_t \rightarrow \one(\pi_t=1)$ for $\delta \rightarrow 0$;  these are the ``realistic individualized treatment rules'' proposed by \textcite{van2007causal,moore2012causal}, which are dynamic but deterministic. %\textcolor{black}{Such rules share an important feature with our proposed interventions: they avoid positivity restrictions by not intervening on those with extreme propensity scores. }
Finally, incremental interventions can recover general stochastic dynamic interventions (where $q^*_t$ replaces the propensity score $\pi_t$) by taking $\delta_t =\{ q^*_t / (1-q^*_t) \} / \{ \pi_t / (1-\pi_t) \}$ for some arbitrary $q^*_t$, whenever defined.

\subsection{Identification}
\label{sec:ident}

In the previous section we described incremental propensity score interventions, which are based on shifting the propensity scores $\pi_t$ by multiplying the odds of receiving treatment by $\delta$. We will now give assumptions that allow for identification of the entire marginal distribution of the resulting potential outcomes $Y^{\bQ(\delta)}$, although for simplicity we focus on estimating just the mean of this distribution.

Importantly, identification of incremental intervention effects requires no conditions on the propensity scores $\pi_t$, \textcolor{black}{since propensity scores that equal zero or one are not shifted}. This is different from more common interventions that require propensity scores to be bounded or otherwise restricted in some way. Specifically we only require the following consistency and exchangeability assumptions.

\begin{assumption}[Consistency]
$Y=Y^{\overline{a}_T}$ if $\overline{A}_T = \overline{a}_T$.
\end{assumption}
\begin{assumption}[Exchangeability]
$A_t \ind Y^{\overline{a}_T} \mid \bH_t$.
\end{assumption}

Consistency means observed outcomes equal corresponding potential outcomes under the observed treatment sequence; it would be violated for example in network settings with interference, where outcomes can be affected by other units' treatment assignment. Exchangeability means treatment assignment is essentially randomized within covariate strata; it can hold by design in a trial, but in observational studies it requires sufficiently many relevant adjustment covariates to be collected.  Importantly, no conditions are needed on the propensity score, since fluctuations based on $q_t$ in \eqref{eq:qdef} will leave the propensity score unchanged if it is zero or one. To the best of our knowledge, the only other work that has discussed removing positivity conditions entirely is \textcite{van2007causal,moore2012causal}; however, they utilize different (deterministic) interventions and consider parametric effect models. \textcolor{black}{General interventions could be modified to similarly avoid positivity, by redefining them to not affect subjects with extreme propensity scores. Two benefits of incremental interventions are (i) avoiding positivity occurs naturally and smoothly via the definition of $q_t$, rather than an inserted indicator; and (ii) as discussed shortly, effects under a wide range of treatment intensities can be summarized with a single curve rather than many regime-specific parameters.}

The next theorem shows that the mean counterfactual outcome $\psi(\delta)=\E(Y^{\bQ(\delta)})$ under the incremental intervention is identified and can be expressed uniquely in terms of the observed data distribution $\Pb$. 

\begin{theorem}
\label{thm:ident}
Under Assumptions 1--2, and if $\delta \in \mathcal{D}=[\delta_\ell, \delta_u]$ for $0 < \delta_\ell \leq \delta_u < \infty$, the incremental effect $\psi(\delta)=\E(Y^{\bQ(\delta)})$ equals
$$\psi(\delta) = \!\! \sum_{\overline{a}_T \in \mathcal{A}^T} \int\limits_{\mathcal{X}} \mu(\bh_T,a_T)  \prod_{t=1}^T \frac{a_t \delta \pi_t(\bh_t) + (1-a_t) \{1-\pi_t(\bh_t)\} }{ \delta \pi_t(\bh_t) + 1- \pi_t(\bh_t) } \ d\Pb(\bx_t \mid \bh_{t-1},a_{t-1})   $$
where $\mathcal{X}=\mathcal{X}_1 \times \dots \times \mathcal{X}_T$ and $\mu(\bh_T,a_T)=\E(Y \mid \bH_T=\bh_T,A_T=a_T)$.
\end{theorem}

Proofs of all theorems are given in the Appendix. Theorem 1 follows from Robins' g-formula \autocite{robins1986new}, replacing the general treatment process under intervention with the proposed incremental intervention $q_t$ indexed by $\delta$. The next corollary shows how the expression for $\psi(\delta)$ simplifies in point exposure studies.

\begin{corollary}
When $T=1$ the identifying expression for $\psi(\delta)$ simplifies to
$$ \psi(\delta) = \E \left[ \frac{ \delta \pi(\bX)  \mu(\bX,1) + \{1-\pi(\bX)\} \mu(\bX,0) }{ \delta \pi(\bX) + \{1 - \pi(\bX)\} } \right]  $$
with $\mu(\bx,a)=\E(Y \mid \bX=\bx, A=a)$. 
\end{corollary}

This corollary shows that, when $T=1$, the incremental effect $\psi(\delta)$ is a weighted average of the regression functions $\mu(\bx,1)$ and $\mu(\bx,0)$, where the weight on $\mu(\bx,1)$ is given by the fluctuated intervention propensity score $q(\bx) = \delta \pi(\bx)/\{ \delta \pi(\bx) + 1-\pi(\bx)\}$ (and the weight on $\mu(\bx,0)$ is $1-q(\bx)$). This weight tends to zero as $\delta \rightarrow 0$ (whenever $\pi(\bx)<1$) and tends to one for $\delta \rightarrow \infty$ (whenever $\pi(\bx)>0$), showing again that $\delta$ controls how far away the intervention is from the observational treatment process. Incremental interventions can range from assigning no one to everyone treatment, but also include an infinite middle ground. Note that we can also write $\psi(\delta) = \E\{\mu(\bX,A^*)\}$ where $A^*$ is a simulated version of treatment under the incremental intervention, with $(A^* \mid \bX=\bx) \sim \text{Bernoulli}\{q(\bx)\}$.

\textcolor{black}{Beyond the fact that identifying incremental effects does not require positivity conditions, targeting $\psi(\delta)$ has another crucial advantage: it is always a one-dimensional curve, regardless of the number of timepoints $T$, and even though it characterizes infinitely many interventions nonparametrically. In contrast, for more traditional causal effects, there is a distinct tension between the number of hypothetical interventions studied and the complexity of the effect. For example one could consider the mean outcome $\E(Y^{\overline{a}_T})$ under all $2^T$ deterministic interventions $\overline{a}_T \in \{0,1\}^T$, but this requires exponentially many parameters without further assumptions. One could impose smoothness across the $2^T$ interventions to reduce the parameter space, but this will yield bias if the smoothness assumptions are incorrect. Conversely, describing the mean outcome under a small number of interventions such as $\overline{a}_T = \mathbf{0}$ and $\overline{a}_T=\mathbf{1}$ (i.e., never treated and always treated) requires only a few parameters, but gives a very limited picture of how changing treatment affects outcomes. In contrast, incremental interventions allow  exploration of infinitely many interventions (one for each $\delta \in \mathcal{D}$), without any parametric assumptions, regardless of how large $T$ is, and still only yield a single curve $\psi: \mathcal{D}  \mapsto \R$ that can be easily visualized with a plot. }

\subsection{Efficiency Theory}
\label{sec:eff}

So far we have introduced incremental propensity score interventions, and showed that resulting effects can be identified without requiring positivity assumptions. Now we will develop general efficiency theory for the incremental effect $\psi(\delta)=\E(Y^{\bQ(\delta)})$.

We refer elsewhere \autocite{bickel1993efficient, van2002semiparametric, van2003unified, tsiatis2006semiparametric, kennedy2016semiparametric} for more detailed information about nonparametric efficiency theory, and so give only a brief review here. A fundamental goal is characterizing so-called influence functions, and in particular finding the efficient influence function. These tasks are essential for a number of reasons. Perhaps most importantly, influence functions can be used to construct estimators with very favorable properties, such as double robustness or general second-order bias (called Neyman orthogonality by \textcite{chernozhukov2016double}). Estimators with these properties can attain fast parametric convergence rates, even in nonparametric settings where nuisance functions are estimated at slower rates via flexible machine learning. The efficient influence function (the only influence function in fully nonparametric models) is particularly important since its variance equals the efficiency bound, thus providing an important benchmark and allowing for the construction of optimal estimators. Influence functions are also critical for understanding the asymptotics of corresponding estimators, since by definition any regular asymptotically linear estimator can be expressed as the empirical average of an influence function plus a negligible $o_p(1/\sqrt{n})$ error term.

Mathematically, influence functions are essentially derivatives. More specifically, viewed as elements of the Hilbert space of mean-zero finite-variance functions, influence functions are those elements whose covariance with parametric submodel scores equals a pathwise derivative of the target parameter. Influence functions also correspond to the derivative in a Von Mises expansion of the target parameter (a distributional analog of a Taylor expansion), and in nonparametric models with discrete support they are a Gateaux derivative of the parameter in the direction of a point mass contamination. 

The result of the next theorem is an expression for the efficient influence function for the incremental effect $\psi(\delta)$ under a nonparametric model, which allows the data-generating process $\Pb$ to be infinite-dimensional. This efficient influence function can be used to characterize the efficiency bound for estimating $\psi(\delta)$, and we will see how this bound changes in randomized trial settings where the propensity scores are known. Then in the next section the efficient influence function will be used to construct estimators, including optimally efficient estimators with the second-order bias property discussed earlier.

\begin{theorem}
\label{thm:eif}
The efficient influence function for $\psi(\delta)$ under a nonparametric model (with unknown propensity scores) is given by
\begin{align*}
%\sum_{t=1}^T  &\left\{ \prod_{s=0}^{t-1} \frac{ (\delta A_s + 1-A_s) }{ \delta \pi_s(\bH_s) + 1 - \pi_s(\bH_s) } \right\} \bigg[ \frac{ \{ \delta m_t(\bH_t,1) - m_t(\bH_t,0) \}\{ \pi_t(\bH_t) - A_t \} }{ \delta \pi_t(\bH_t) + 1 - \pi_t(\bH_t) } \\
%& + \frac{ \delta \{ m_t(\bH_t,1) - m_t(\bH_t,0) \}\{ A_t - \pi_t(\bH_t) \} }{\{  \delta \pi_t(\bH_t) + 1 - \pi_t(\bH_t) \}^2 } \bigg] + \prod_{t=1}^{T} \frac{ (\delta A_t + 1-A_t) Y }{ \delta \pi_t(\bH_t) + 1 - \pi_t(\bH_t) }  - \psi(\delta)
\sum_{t=1}^T  &  \left[ \frac{A_t \{1-\pi_t(\bH_t)\} - (1-A_t) \delta \pi_t(\bH_t) }{ \delta/(1-\delta)} \right] \left[ \frac{ \delta \pi_t(\bH_t) m_t(\bH_t,1) + \{1-\pi_t(\bH_t)\} m_t(\bH_t,0) } { \delta \pi_t(\bH_t) + 1 - \pi_t(\bH_t) } \right] \\
& \hspace{.1in} \times \left\{ \prod_{s=1}^{t} \frac{ (\delta A_s + 1-A_s) }{ \delta \pi_s(\bH_s) + 1 - \pi_s(\bH_s) } \right\} + \prod_{t=1}^{T} \frac{ (\delta A_t + 1-A_t) Y }{ \delta \pi_t(\bH_t) + 1 - \pi_t(\bH_t) }  - \psi(\delta)
\end{align*}
where for $t=0,...,T-1$ we define 
\begin{align*} m_t(\bh_t,a_t) &=  \int_{\mathcal{R}_t} \mu(\bh_T,a_T) \prod_{s={t+1}}^T  \frac{a_s \delta \pi_s(\bh_s) + (1-a_s) \{1-\pi_s(\bh_s)\} }{ \delta \pi_s(\bh_s) + 1- \pi_s(\bh_s) } \ d\Pb(\bx_{s} \mid \bh_{s-1},a_{s-1}) %\\
%&= \E\left[  \frac{ \delta \pi_{t+1}(\bH_{t+1}) m_{t+1}(\bH_{t+1},1) + \{1-\pi_{t+1}(\bH_{t+1})\} m_{t+1}(\bH_{t+1},0) }{ \delta \pi_{t+1}(\bH_{t+1}) + 1- \pi_{t+1}(\bH_{t+1}) } \Bigm| \bH_t=\bh_t, A_t=a_t \right]
\end{align*}
with $\mathcal{R}_t= (\mathcal{H}_T \times \mathcal{A}_T) \setminus \mathcal{H}_t$, and for $t=T$ we let $m_T(\bh_T,a_T)=\mu(\bh_T,a_T)$.
\end{theorem}

We give a proof of Theorem \ref{thm:eif} in Section \ref{sec:proof_eif} of the Appendix, by way of deriving the efficient influence function for general stochastic interventions with treatment distributions that depend on the observational propensity scores. To the best of our knowledge this result has not yet appeared in the literature, and will be useful for general stochastic interventions beyond those with the incremental form proposed here, \textcolor{black}{regardless of whether they depend on the observational treatment process or not. Our result recovers previously proposed influence functions for other stochastic intervention effects in the $T=1$ setting as special cases \autocite{diaz2012population, haneuse2013estimation},  and could be used to generalize this work to the multiple timepoint setting. Further, our result can also be used to construct the efficient influence function and corresponding estimator for other stochastic intervention effects, for which there are currently only likelihood-based and weighting estimators available \autocite{moore2012causal,young2014identification}.} 

The structure of the efficient influence function in Theorem \ref{thm:eif} is somewhat similar to that of more standard effect parameters, in the sense that it consists of an inverse-probability-weighted term (the rightmost product term in the second line) as well as an augmentation term. However the particular form of the weighted and augmentation terms are quite different from those that appear in more common causal and missing data problems. We discuss the weighted term in more detail in Section \ref{sec:simple}, when we introduce an inverse-probability-weighted estimator for $\psi(\delta)$. The augmentation term involves the functions $m_t$, which  can be viewed as marginalized versions of the full regression function $\mu(\bh_t,a_t)$ that conditions on all of the past (with smaller values of $t$ coinciding with more marginalization). 

Note that for notational simplicity we drop the dependence of $m_t$ on $\delta$ and $(\pi_{t+1},...,\pi_T)$, as well as on the conditional densities of the covariates $(\bX_{t+1},...,\bX_T)$. Importantly, the pseudo-regression functions $m_t$ also have a recursive sequential regression formulation, as displayed in the subsequent remark. 

\begin{remark}
The functions $m_t$ can be equivalently expressed recursively as
$$ m_{t-1}(\bH_{t-1},A_{t-1}) = \E\left[  \frac{ \delta \pi_{t}(\bH_{t}) m_{t}(\bH_{t},1) + \{1-\pi_{t}(\bH_{t})\} m_{t}(\bH_{t},0) }{ \delta \pi_{t}(\bH_{t}) + 1- \pi_{t}(\bH_{t}) } \Bigm| \bH_{t-1}, A_{t-1} \right] $$
for  $t=1,...,T$ and $m_T(\bh_T,a_T)=\mu(\bh_T,a_T)$ as before.
\end{remark}

Viewing the $m_t$ functions in the above sequential regression form is very practically useful for the purposes of estimation. Specifically it shows how to bypass conditional density estimation, and instead construct estimates $\hat{m}_t$ using regression methods that are more commonly found in statistical software. 

It is also important to note that the pseudo-regressions $m_t$ depend on the observational treatment process; this is not the case for analogous influence function terms for more common parameters like $\E(Y^{\overline{a}_T})$. This is due to the fact that the functional $\psi(\delta)$ itself depends on the observational treatment process, which means for example that double robustness is not possible (though second-order bias still is) and that the efficiency bound is different when the propensity scores are known versus unknown. The issue of double robustness is discussed in more detail in Section \ref{sec:clt}. In Lemmas \ref{lem:eif_qknown} and \ref{lem:eif_contribution} in the Appendix we give the efficient influence function when the propensity scores are known, as well as a specific expression for the contribution that comes from the scores being unknown, both for general (possibly non-incremental) stochastic interventions. 

In the next corollary we give the efficient influence function for the incremental effect in a single timepoint study, which has a simpler and more intuitive form.

\begin{corollary}
When $T=1$ the efficient influence function for $\psi(\delta)$ simplifies to
\begin{align*}
&\frac{ \delta \pi(\bX)  \phi_1(\bZ) + \{1-\pi(\bX)\} \phi_0(\bZ) }{ \delta \pi(\bX) + \{1 - \pi(\bX)\} } + \frac{ \delta \gamma(\bX) \{ A - \pi(\bX)\} }{ \{ \delta \pi(\bX) + 1 - \pi(\bX) \}^2 }  - \psi(\delta)% \\
%&\hspace{.2in}= \frac{ (\delta A + 1-\delta)Y }{\delta \pi(\bX) + 1-\pi(\bX)}  - \{A-\pi(\bX)\} \frac{ \delta \mu(\bX,1)-\mu(\bX,0)}{\delta \pi(\bX) + 1-\pi(\bX)} 
\end{align*}
where $\gamma(\bx)=\mu(\bx,1)-\mu(\bx,0)$ and 
$$\phi_a(\bZ)=\frac{\one(A=a)}{\pi(a \mid \bX)} \{Y-\mu(\bX,a)\} + \mu(\bX,a)$$ 
is the uncentered efficient influence function for the parameter $\E\{\phi_a(\bZ)\} = \E\{\mu(\bX,a)\}$.
\end{corollary}

The efficient influence function in the $T=1$ case is therefore a simple weighted average of the influence functions for $\E(Y^1)$ and $\E(Y^0)$, plus a contribution that comes from the fact that the propensity score is unknown and must be estimated. If the propensity scores were known, the efficient influence function would just be the first weighted average term in Corollary 2. As will be discussed in more detail in the next section, estimating the influence function in the $T=1$ case is straightforward as it only depends on the regression function $\mu$ and propensity score $\pi$ (rather than the sequential psuedo-regression functions $m_t$ that appear in the longitudinal setting).

\section{Estimation \& Inference}
\label{sec:estimation}

In this section we develop estimators for the proposed incremental effect $\psi(\delta)$. We focus our analysis on flexible sample-splitting estimators that allow arbitrarily complex nuisance estimation, e.g., via high-dimensional regression and machine learning methods; however we also discuss simpler estimators that rely on empirical process conditions to justify full-sample nuisance estimation. In particular we show that there exists an inverse-probability-weighted estimator of the incremental effect that is especially easy to compute. We go on to describe the asymptotic behavior of our proposed estimators, both from a pointwise perspective and uniformly across a continuum of increment parameter $\delta$ values. Finally we propose a computationally efficient multiplier-bootstrap approach for constructing uniform confidence bands across $\delta$, and use it to develop a novel test of no treatment effect.

\subsection{Simple Estimators}
\label{sec:simple}

We first describe various simple estimators of the incremental effect, which provide some intuition for the main estimator we propose in the next section. The simple inverse-probability-weighted estimator discussed here might be preferred if the propensity scores can be modeled well (e.g., in a randomized trial) and computation comes at a high cost.

Let $\varphi(\bZ; \boldsymbol\eta, \delta)$ denote the (uncentered) efficient influence function from Theorem 2, which is a function of the observations $\bZ$ and the nuisance functions
$$ \boldsymbol\eta = (\boldsymbol\pi, \textbf{m}) = (\pi_1,...,\pi_T, m_1,...,m_T) . $$
By uncentered we mean that $\varphi(\bZ;\boldsymbol\eta,\delta)$ equals the quantity displayed in Theorem 2 plus the parameter $\psi(\delta)$, so that $\E\{\varphi(\bZ;\boldsymbol\eta,\delta)\}=\psi(\delta)$ by construction.

If one is willing to rely on appropriate empirical process conditions (e.g., Donsker-type or low entropy conditions, as discussed by \textcite{van1996weak}, \textcite{van2000asymptotic}, and others) then a natural estimator would be given by the solution to the efficient influence function estimating equation, i.e., the Z-estimator
$$ \hat\psi^*(\delta) = \Pn\{ \varphi(\bZ;\hat{\boldsymbol\eta}, \delta) \} $$
where $\hat{\boldsymbol\eta}$ are some initial estimators of the nuisance functions, and $\Pn$ denotes the empirical measure so that sample averages can be written as $\frac{1}{n} \sum_i f(\bZ_i) = \Pn\{f(\bZ)\} =  \int f(\bz)  \ d\Pn(\bz)$. An algorithm describing how to compute the estimator $\hat\psi^*(\delta)$ is given in Section \ref{sec:algorithm} of the Appendix. As a special case, if the propensity scores $\pi_t$ can be correctly modeled parametrically (e.g., when they are known as in a randomized trial) then one could use the simple inverse-probability-weighted estimator given by
$$ \hat\psi_{ipw}^*(\delta) = \Pn\left\{ \prod_{t=1}^T \frac{( \delta A_t + 1-A_t) Y}{\delta \hat\pi_t(\bH_t) + 1-\hat\pi_t(\bH_t)} \right\} . $$
This estimator can be computed very quickly, as it only requires fitting a single pooled regression to estimate $\pi_t$ and then taking a weighted average. However it has disadvantages, as will be discussed shortly.  Also note that it is a special case of $\hat\psi^*(\delta)$ that sets $\hat{m}_t=0$. 

It is instructive to compare the inverse-probability-weighted estimator above to that for a usual deterministic static intervention effect like $\E(Y^{\overline{a}_T})$. For example, the inverse-probability-weighted estimator of the quantity $\E(Y^\mathbf{1})$ weights each always-treated unit by the (inverse) product of propensity scores $\prod_t \hat\pi_t$, and otherwise assigns zero weight. In contrast, when $\delta>1$ the estimator $\hat\psi_{ipw}^*(\delta)$ weights each treated time by the (inverse of the) propensity score plus some fractional contribution of its complement, i.e., $\hat\pi_t + (1-\hat\pi_t)/\delta$, where the size of the contribution decreases with $\delta$; untreated times are weighted by this same amount, except the entire weight is further downweighted by a factor of $\delta$. Therefore when $\delta$ is very large, the two inverse-probability-weighted estimators coincide. However, for cases when $\delta$ is not very large, this also indicates why the estimator $\hat\psi_{ipw}^*(\delta)$ is immune to extreme weights: even if $\hat\pi_t$ is very small, there will still be a contribution to the weight that moves it away from zero.

\subsection{Proposed Estimator}
\label{sec:proposed}

Although the estimators presented in the previous section are relatively simple, they have some disadvantages. First, the inverse-probability-weighted estimator $\hat\psi_{ipw}^*(\delta)$ will in general not be $\sqrt{n}$-consistent unless all the propensity scores are estimated with correctly specified parametric models; this is typically an unreasonable assumption outside of randomized trials where propensity scores are known. In point exposure studies with a single timepoint, (saturated) parametric models might be used if the adjustment covariates are low-dimensional. However, in studies with more than just a few timepoints, the histories $\bH_t$ can easily be high-dimensional even if the covariates $\bX_t$ are low-dimensional, making parametric modeling assumptions less tenable even in the low-dimensional $\bX_t$ case.

In contrast, the more general Z-estimator $\hat\psi^*(\delta)$ can converge at fast parametric $\sqrt{n}$ rates (and attain the efficiency bound from Section \ref{sec:eff}), even when the propensity scores $\pi_t$ and pseudo-outcome regressions $m_t$ are modeled flexibly and estimated at rates slower than $\sqrt{n}$, as long as these nuisance functions are estimated consistently at rates faster than $n^{1/4}$. Lowering the bar from $\sqrt{n}$ to $n^{1/4}$ for the nuisance estimator convergence rate allows much more flexible nonparametric methods to be employed; for example these rates are attainable under smoothness, sparsity, or other nonparametric structural constraints. However, as mentioned earlier, these Z-estimator properties require some empirical process conditions that restrict the flexibility and complexity of the nuisance estimators. This is essentially because $\hat\psi^*(\delta)$ uses the sample twice, once for estimating the nuisance functions $\boldsymbol\eta$ and again for evaluating the influence function $\varphi$. Without restricting the entropy of the nuisance estimators, using the full sample in this way can result in overfitting and intractable asymptotics. Unfortunately, the required empirical process conditions may not be satisfied by many modern regression  methods, such as random forests, boosting, deep learning, or complicated ensembles. 

In order to accommodate the added complexity of these modern machine learning tools, we use sample splitting \autocite{zheng2010asymptotic, chernozhukov2016double}. This avoids the problematic ``double'' use of the sample and, as will be seen in the next section, yields asymptotically normal and efficient estimators without any restrictions on the complexity of the nuisance estimators (however, $n^{1/4}$-type rate conditions are still required).

Therefore we randomly split the observations $(\bZ_1,...,\bZ_n)$ into $K$ disjoint groups, using a random variable $S$ drawn independently of the data, where $S_i \in \{1,...,K\}$ denotes the group membership for unit $i$. Then our proposed estimator is given by
$$ \hat\psi(\delta) = \frac{1}{K} \sum_{k=1}^K \Pn^k \{ \varphi(\bZ;\hat{\boldsymbol\eta}_{\text{-}k}, \delta) \} = \Pn \{ \varphi(\bZ; \hat{\boldsymbol\eta}_{\text{-}S}, \delta)\}$$
where we let $\Pn^k$ denote empirical averages only over the set of units $\{i : S_i = k\}$ in group $k$ (i.e., $\Pn\{f(\bZ)\} = \sum_i f(\bZ_i) \one(S_i=k) / \sum_i \one(S_i=k) $), and we let $\hat{\boldsymbol\eta}_{\text{-}k}$ denote the nuisance estimator constructed excluding group $k$, i.e., only using those units $\{i : S_i \neq k\}$ in groups $\mathcal{K} \setminus k$.  It is hoped that $\hat{\boldsymbol\eta}_{\text{-}k}$ is a rate-optimal estimator of the nuisance functions, for example constructed using kernels, splines, penalized regression, boosting, random forests, etc., or some ensemble-based combination.

An algorithm detailing exactly how to compute the estimator $\hat\psi(\delta)$ is given as follows. For reference, the algorithm for the non-sample splitting estimator $\hat\psi^*(\delta)$ is also given in Section \ref{sec:algorithm} of the Appendix and contains the main ideas.

\begin{algorithm} For each $\delta$ and $k$, letting $\bD_0=\{\bZ_i : S_i \neq k\}$ and $\bD_1=\{\bZ_i : S_i=k\}$ denote corresponding training and test data, respectively, and $\bD=\bD_0 \cup \bD_1$:
\begin{enumerate} 
\item Regress $A_t$ on $\bH_t$ in $\bD_0$, obtain predicted values $\hat\pi_t(\bH_t)$ for each subject/time in $\bD$.
\item Construct time-dependent weights $W_t = \frac{\delta A_t + 1-A_t}{\delta \hat\pi_t(\bH_t) + 1-\hat\pi_t(\bH_t)}$ in $\bD_1$ for each subject/time.
\item Calculate cumulative product weight $\widetilde{W}_t = \prod_{s=1}^t W_s$ in $\bD_1$ for each subject/time.
\item For each time $t=T, T-1, ..., 1$ (starting with $R_{T+1}=Y$):
\begin{enumerate}
\item Regress $R_{t+1}$ on $(\bH_t,A_t)$ in $\bD_0$, obtain predictions $\hat{m}_t(\bH_t,1)$, $\hat{m}_t(\bH_t,0)$ in $\bD$.
\item Construct pseudo-outcome $R_t=\frac{ \delta \hat\pi_t(\bH_t) \hat{m}_t(\bH_t,1) + \{1- \hat\pi_t(\bH_t)\} \hat{m}_t(\bH_t,0) } { \delta \hat\pi_t(\bH_t) + 1 - \hat\pi_t(\bH_t) }$ in $\bD$.
\end{enumerate}
\item Compute time-dependent weights $V_t = \frac{A_t \{1- \hat\pi_t(\bH_t)\} - (1-A_t) \delta \hat\pi_t(\bH_t) }{ \delta/(1-\delta)}$ in $\bD_1$.
\item Compute $\varphi = \widetilde{W}_T Y + \sum_t \widetilde{W}_t V_t R_t $ in $\bD_1$ and define $\hat\psi_k(\delta)$ to be its average in $\bD_1$.
\end{enumerate}
Finally, set $\hat\psi(\delta)$ to be the average of the $K$ estimators $\hat\psi_k(\delta)$, $k=1,...,K$.
\end{algorithm}

Importantly, computing $\hat\psi(\delta)$ only requires estimating regression functions (e.g., using random forests) and not conditional densities, due to the recursive regression formulation of the functions $m_t$ in Remark 2. Although the process can be somewhat computationally expensive depending on the number of timepoints $T$, sample size $n$, and grid density for $\delta$, it is easily parallelizable due to the sample splitting. For a single timepoint all estimators are easy and fast to compute. In Section \ref{sec:code} of the Appendix, we provide a user-friendly R function for general use in cross-sectional or longitudinal studies; \textcolor{black}{the function can also be found in the} \verb|npcausal| \textcolor{black}{R package available at GitHub (github.com/ehkennedy/npcausal).}

\subsection{Weak Convergence}
\label{sec:clt}

In this section we detail the main large-sample property of our proposed estimator, that $\hat\psi(\delta)$ is $\sqrt{n}$-consistent and asymptotically normal under weak conditions (mostly only requiring that the nuisance functions are estimated at faster than $n^{1/4}$ rates). This result holds both pointwise for a given $\delta$, and uniformly in the sense that, after scaling and when viewed as a random function on $\mathcal{D}=[\delta_\ell,\delta_u]$, the estimator converges in distribution to a Gaussian process. The latter fact is crucial for developing uniform confidence bands, as well as the test of no treatment effect we present in the next section. \textcolor{black}{Importantly, the estimator attains fast $\sqrt{n}$ rates even under nonparametric assumptions and even though the target parameter is a curve; this is often not possible \autocite{kennedy2016robust, kennedy2017nonparametric}. }

In what follows we denote the squared $L_2(\Pb)$ norm by $\|f\|^2 = {\int f(\bz)^2 \ d\Pb(\bz)}$. When necessary, we depart slightly from previous sections and index the pseudo-regression functions $m_{t,\delta}$ (and their estimators $\hat{m}_{t,\delta}$) by both time $t$ and the increment parameter $\delta$. The next result lays the foundation for our proposed inferential and testing procedures.

\begin{theorem}
\label{thm:unif_clt}
Let $\hat\sigma^2(\delta) = \Pn[\{\varphi(\bZ;\hat{\boldsymbol\eta}_{\text{-}S},\delta)-\hat\psi(\delta)\}^2]$ denote the estimator of the variance function $\sigma^2(\delta) = \E[\{\varphi(\bZ;\boldsymbol\eta,\delta)-\psi(\delta)\}^2]$. Assume: 
\begin{enumerate}
\item The set $\mathcal{D} =[\delta_\ell,\delta_u]$ is bounded with $0<\delta_\ell \leq \delta_u < \infty$.
\item $\Pb\{ |m_t(\bH_t,A_t) | \leq C\} =\Pb\{ |\hat{m}_t(\bH_t,A_t) | \leq C\} = 1$ for some $C < \infty$ and all $t$.
\item $\sup_{\delta \in \mathcal{D}} | \frac{\hat\sigma(\delta)}{\sigma(\delta)} - 1| = o_\Pb(1)$, and $\| \sup_{\delta \in \mathcal{D}} | \varphi(\bz; \hat{\boldsymbol\eta},\delta) - \varphi(\bz; \boldsymbol\eta,\delta)| \ \! \| = o_\Pb(1)$.
\item $\Big( \sup_{\delta \in \mathcal{D}}  \| \hat{m}_{t,\delta} - m_{t,\delta} \| + \| \hat\pi_t - \pi_t \| \Big) \| \hat\pi_s - \pi_s \| = o_\Pb(1/\sqrt{n})$ for $s \leq t \leq T$.
\end{enumerate}
Then
$$  \frac{ \hat\psi(\delta) - \psi(\delta) }{\hat\sigma(\delta) / \sqrt{n} }  \indist  \Gb(\delta)  $$
in $\ell^\infty(\mathcal{D})$, where $\Gb(\cdot)$ is a mean-zero Gaussian process with covariance  $\E\{\Gb(\delta_1) \Gb(\delta_2)\} = \E\{ \widetilde\varphi(\bZ;\boldsymbol\eta,\delta_1) \widetilde\varphi(\bZ;\boldsymbol\eta,\delta_2) \}$ and $\widetilde\varphi(\bz;\boldsymbol\eta,\delta)=\{\varphi(\bz;\boldsymbol\eta,\delta)-\psi(\delta)\}/\sigma(\delta)$.
\end{theorem}

The proof of Theorem \ref{thm:unif_clt} is given in Section \ref{proof:unif_clt} of the Appendix. The logic of the proof is roughly similar to that used by \textcite{belloni2015uniformly}, but we avoid their restrictions on nuisance function entropy by sample-splitting and arguing conditionally on the training data. This allows for the use of arbitrarily complex estimators $\hat{\boldsymbol\eta}$, such as random forests, boosting, etc.  We also do not need explicit smoothness assumptions on $\psi(\delta)$ or $\varphi(\bZ;\boldsymbol\eta, \delta)$ since they are necessarily Lipschitz in $\delta$ by construction, based on our choice of the incremental intervention distribution $q_t$.

Assumptions 1--2 of Theorem \ref{thm:unif_clt} are mild boundedness conditions on the set $\mathcal{D}$ of $\delta$ values and the functions $m_t$ and their estimators, respectively. Assumption 2 could be relaxed at the expense of a less simple proof, for example with bounds on $L_p$ norms. Assumption 3 is a basic and mild consistency assumption, with no requirement on rates of convergence. The main substantive assumption is Assumption 4, which says the nuisance estimators must be consistent and converge at a fast enough rate (essentially $n^{1/4}$ in $L_2$ norm).

Importantly, the rate condition in Assumption 4 can be attained under nonparametric smoothness, sparsity, or other structural constraints. We are agnostic about how such rates might be attained since the particular required assumptions are problem-dependent; in practice we suggest using ensemble learners that can adapt to diverse kinds of structure. The particular form of the rate requirement indicates that double robustness is not possible, since we need products of the form $\| \hat\pi_t - \pi_t \| \| \hat\pi_s - \pi_s \|$ to be small, thus requiring consistent estimation of the propensity scores (albeit only at slower than parametric rates). If the propensity scores are known as in a randomized trial, then Assumption 4 will necessarily hold; in this case, the result of the theorem follows with $\varphi(\bz;\boldsymbol\eta,\delta)$ evaluated at $\overline{m}_t$ the limit of the estimator $\hat{m}_t$, which may or may not equal the true pseudo-regression $m_t$. If the propensity scores are estimated with correct parametric models, then Assumption 4 would only require a (uniformly) consistent estimator of $m_t$, without any rate conditions.

Based on the result in Theorem \ref{thm:unif_clt}, pointwise 95\% confidence intervals for $\psi(\delta)$ can be constructed as
$$ \hat\psi(\delta) \pm 1.96 \ \hat\sigma(\delta) / \sqrt{n} $$
where $\hat\sigma^2(\delta) = \Pn[\{\varphi(\bZ;\hat{\boldsymbol\eta}_{\text{-}S},\delta)-\hat\psi(\delta)\}^2]$ is the variance estimator given in the statement of the theorem. Uniform inference and testing is discussed in the next section.

\subsection{Uniform Inference \& Testing No Effect}
\label{sec:test}

In this section we present a multiplier bootstrap approach to obtaining uniform confidence bands for the incremental effect curve $\{\psi(\delta) : \delta \in \mathcal{D} \}$, along with a corresponding novel test of no treatment effect. This test can be useful in general causal inference problems, even when positivity assumptions are justified and even if incremental effects are not of particular interest.

To construct a $(1-\alpha)$ uniform confidence band of the form $ \hat\psi(\delta) \pm c_\alpha \hat\sigma(\delta) / \sqrt{n} $, as usual we need to find a critical value $c_\alpha$ that satisfies
$$ %\Pb \left\{ \hat\psi(\delta) - \frac{ {c}_\alpha \hat\sigma(\delta)}{\sqrt{n} } \leq \psi(\delta) \leq \hat\psi(\delta) +  \frac{ {c}_\alpha \hat\sigma(\delta)}{\sqrt{n} }, \forall \delta \in \mathcal{D} \right\}  = 
\Pb\left( \sup_{\delta \in \mathcal{D}} \left| \frac{\hat\psi(\delta) - \psi(\delta)}{\hat\sigma(\delta)/\sqrt{n}} \right| \leq {c}_\alpha \right) = 1-\alpha + o(1), $$
since the expression on the left is the probability that the band covers the true incremental effect curve $\psi(\delta)$ for all $\delta \in \mathcal{D}$. 

Based on the result of Theorem \ref{thm:unif_clt}, this critical value can be obtained by approximating the distribution of the supremum of the Gaussian process $\{\Gb(\delta) : \delta \in \mathcal{D}\}$ with covariance function as given in the statement of the theorem. We use the multiplier bootstrap \autocite{gine1984some, van1996weak, belloni2015uniformly} to approximate this distribution. A primary advantage of the multiplier bootstrap is its computational efficiency, since it does not require refitting the nuisance estimators, which can be expensive when there are many covariates and/or timepoints.

The idea behind the multiplier bootstrap is to approximate the distribution of the aforementioned supremum with the supremum of the  multiplier process
$$ \sqrt{n} \Pn \Big[ \xi \{ \varphi(\bZ;\hat{\boldsymbol\eta}_{\text{-}S},\delta)- \hat\psi(\delta) \} / \hat\sigma(\delta) \Big] $$
over draws of the multipliers $(\xi_1,...,\xi_n)$ (conditional on the sample data $\bZ_1,...,\bZ_n$), which are iid random variables with mean zero and unit variance that are independent of the sample. Typically one uses either Gaussian or Rademacher multipliers (i.e., $\Pb(\xi=1)=\Pb(\xi=-1)=0.5$); we use Rademacher multipliers because they gave better performance in simulations. The next theorem states that this approximation works under the same assumptions from Theorem \ref{thm:unif_clt}.

\begin{theorem}
\label{thm:unif_ci}
Let $\hat{c}_\alpha$ denote the $1-\alpha$ quantile (conditional on the data) of the supremum of the multiplier bootstrap process, i.e.,
$$ \Pb\left( \sup_{\delta \in \mathcal{D}} \left| \sqrt{n} \ \Pn\!\left[ \xi \left\{ \frac{\varphi(\bZ;\hat{\boldsymbol\eta}_{\text{-}S},\delta)- \hat\psi(\delta)}{\hat\sigma(\delta)}  \right\} \right] \right| \geq \hat{c}_\alpha \Bigm| \bZ_1,...,\bZ_n \right) = \alpha $$
where $(\xi_1,...,\xi_n)$ are iid Rademacher random variables independent of the sample. Then, under the same conditions from Theorem \ref{thm:unif_clt},
$$ \Pb \left\{ \hat\psi(\delta) - \frac{ \hat{c}_\alpha \hat\sigma(\delta)}{\sqrt{n} } \leq \psi(\delta) \leq \hat\psi(\delta) +  \frac{ \hat{c}_\alpha \hat\sigma(\delta)}{\sqrt{n} }, \text{ for all } \delta \in \mathcal{D} \right\} = 1 - \alpha + o(1) . $$
\end{theorem}

The proof of Theorem \ref{thm:unif_ci} is given in Section \ref{proof:unif_ci} of the Appendix, and follows by linking the multiplier bootstrap process to the same Gaussian process $\Gb$ to which the scaled estimator $\hat\psi(\delta)$ converges. As mentioned above, the multiplier bootstrap only requires simulating the multipliers $\xi$ and not re-estimating the nuisance functions, so it is straightforward and fast to implement. We include an implementation in the R function given in Section \ref{sec:code} of the Appendix, \textcolor{black}{as well as in the} \verb|npcausal| \textcolor{black}{R package available at GitHub (github.com/ehkennedy/npcausal).}

Given the above uniform confidence band, we can test the null hypothesis of no incremental intervention effect
$$ H_0: \psi(\delta) = \E(Y) \ \text{ for all } \delta \in \mathcal{D} , $$
by simply checking whether a $(1-\alpha)$ band contains a straight line over $\mathcal{D}$. In other words we can compute a p-value as
$$ \hat{p} = \sup \Big\{ \alpha : \inf_{\delta \in \mathcal{D}} \{ \hat\psi(\delta) + \hat{c}_\alpha \hat\sigma(\delta) / \sqrt{n} \} \geq \sup_{\delta \in \mathcal{D}} \{ \hat\psi(\delta) - \hat{c}_\alpha \hat\sigma(\delta) / \sqrt{n} \} \Big\} . $$
Note that the condition in the above set corresponds to failing to reject $H_0$ at level $\alpha$, since there is space for a straight line between the smallest upper confidence limit and largest lower confidence limit. We will necessarily fail to reject at level $\alpha=0$ since this amounts to an infinitely wide confidence band, and the p-value is the largest $\alpha$ at which we fail to reject (i.e., the p-value is small if we reject even for wide bands, and large if we need to move to narrower bands or never reject).

Interestingly, the hypothesis we test above lies in a middle ground between Fisher's null of no individual effect and Neyman's null of no average effect. $H_0$ is a granular hypothesis perhaps closer to Fisher's null than Neyman's, but it can still be tested nonparametrically and in a longitudinal superpopulation framework. This is in contrast to common tests of Fisher's null that operate under additive effect hypotheses and are limited to point exposures \autocite{rosenbaum2002covariance}. Thus tests of the null $H_0$ can be useful in general settings, independent of any interest in pursuing incremental intervention effects or avoiding positivity assumptions.

\section{Illustrations}

\subsection{Simulation Study}

Here we explore finite-sample properties via simulation, based on the simulation setup used by \textcite{kang2007demystifying}. In particular we consider their model
\begin{equation*}
\begin{gathered}
(X_1,X_2,X_3,X_4) \sim N(\mathbf{0},\mathbf{I}) , \\
\Pb(A=1 \mid \bX)= \expit(-X_1 + 0.5 X_2 - 0.25 X_3 - 0.1 X_4)  \\
(Y \mid \bX, A) \sim N\{ \mu(\bX,A) , 1 \} 
\end{gathered}
\end{equation*}
where the regression function is given by $\mu(\bx,a) = 200 + a\{10 + 13.7(2x_1 + x_2 + x_3 + x_4)\}$. 
This simulation setup is known to yield variable propensity scores that can degrade the performance of weighting-based estimators. 

We considered three estimators in our simulation: a plug-in estimator given by
$$ \hat\psi_{pi}(\delta) = \Pn \left[ \frac{\delta \hat\pi(\bX) \hat\mu(\bX,1) + \{1-\hat\pi(\bX)\} \hat\mu(\bX,0) }{\delta \hat\pi(\bX) + 1-\hat\pi(\bX)} \right] , $$
along with the inverse-probability-weighted (IPW) estimator and proposed efficient estimator described in Sections \ref{sec:simple}--\ref{sec:proposed}. We further considered four versions of each these estimators, depending on how the nuisance functions were estimated: correct parametric models, misspecified parametric models based on transformed covariates $\bX^*$ 
%$$ \bX^* = \Big[ \exp(X_1/2) , 10 + X_2 / \{1+\exp(X_1) \} ,  ( X_1 X_3 / 25  + 0.6 )^3 , (X_2 + X_4 + 20)^2 \Big]^\T $$
(using the same covariate transformations as \textcite{kang2007demystifying}), and nonparametric estimation (using original or transformed covariates). For nonparametric estimation we used the cross-validation-based Super Learner ensemble \autocite{van2007super} to combine generalized additive models, multivariate adaptive regression splines, support vector machines, and random forests, along with parametric models (with and without interactions, and with terms selected stepwise via AIC). Regardless of estimator (plug-in, IPW, or proposed), for nonparametric nuisance estimation we used sample splitting as described in Section \ref{sec:proposed} with $K=2$ splits.

Estimator performance was assessed via integrated bias and root-mean-squared error
$$ \widehat{\text{bias}} = \frac{1}{I} \sum_{i=1}^I \Big| \frac{1}{J} \sum_{j=1}^J \hat\psi_j(\delta_i) - \psi(\delta_i) \Big| \ , \ \ 
\widehat{\text{RMSE}} = \frac{\sqrt{n}}{I} \sum_{i=1}^I \left[ \frac{1}{J} \sum_{j=1}^J \Big\{ \hat\psi_j(\delta_i) - \psi(\delta_i) \Big\}^2 \right]^{1/2} $$
across $J=500$ simulations and $I=100$ values of $\delta$ equally spaced (on the log scale) between $\exp(-2.3) \approx 0.1$ and $\exp(2.3) \approx 10$. Results are given in Figure \ref{tab:simtab}.

\begin{figure}[h!]
\begin{center}
\includegraphics[width=\textwidth]{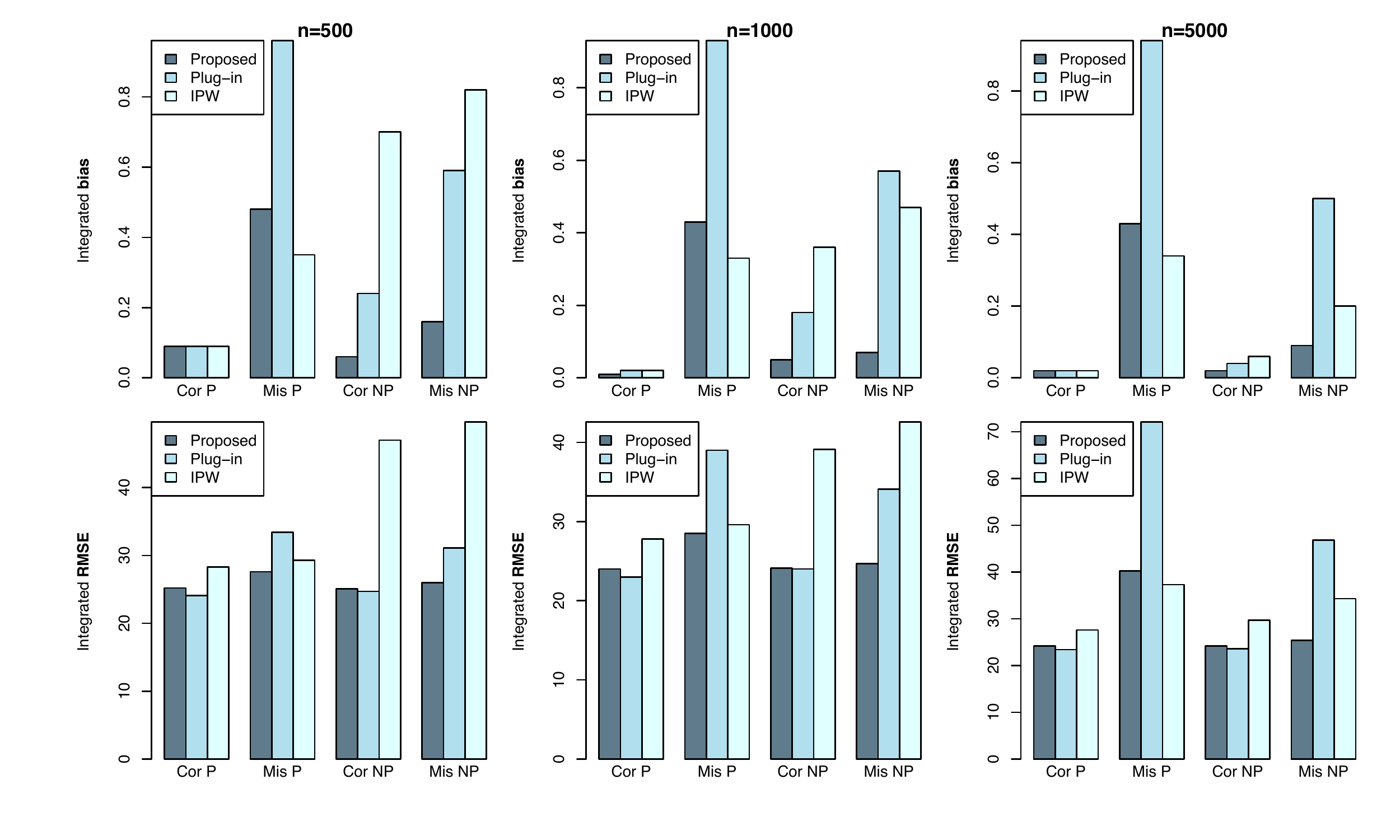}
\caption{Integrated bias and root-mean-squared-error (RMSE) across 500 simulations. (IPW = inverse-probability-weighted; 
P/NP = parametric/nonparametric nuisance estimation based on covariates $\bX$ (Cor) or transformed version $\bX^*$ (Mis).)} \label{tab:simtab}
\end{center}
\end{figure}

\begin{comment}
\begin{table}[h!]
\caption{Integrated bias and RMSE across 500 simulations. \label{tab:simtab}}
\begin{center}
\begin{tabular}{llrrrr}
\hline
Sample & & \multicolumn{4}{c}{Integrated bias (RMSE) for setting:} \\
size $n$ & Method & Cor P & Mis P & Cor NP & Mis NP  \\
\hline
100 & Plug-in & 0.24 (24.9) & 0.94 (27.1)& 0.80 (25.5) & 1.34 (30.9)  \\
& IPW & 0.22 (29.5) & 0.42 (29.1) & 5.08 (99.9) & 4.79 (97.3)  \\
& Proposed & 0.23 (26.1) & 0.61 (26.8) & 0.26 (27.2) & 0.16 (27.7) \\
\\
500 & Plug-in & 0.09 (24.1) & 0.96 (33.4) & 0.24 (24.7) & 0.59 (31.1) \\
& IPW & 0.09 (28.3) & 0.35 (29.3) & 0.70 (47.0) & 0.82 (49.7) \\
& Proposed & 0.09 (25.2) & 0.48 (27.6) & 0.06 (25.1) & 0.16 (26.0) \\
\\
1000 & Plug-in & 0.02 (23.0) & 0.93 (39.0) & 0.18 (24.0) & 0.57 (34.1) \\
& IPW & 0.02 (27.8) & 0.33 (29.6) & 0.36 (39.1) & 0.47 (42.6) \\
& Proposed & 0.01 (24.0) & 0.43 (28.5) & 0.05 (24.1) & 0.07 (24.7) \\
\\
5000 & Plug-in & 0.02 (23.4) & 0.94 (72.1) & 0.04 (23.6) & 0.50 (46.8) \\
& IPW & 0.02 (27.6) & 0.34 (37.3) & 0.06 (29.7) & 0.20 (34.3) \\
& Proposed & 0.02 (24.2) & 0.43 (40.2) & 0.02 (24.2) & 0.09 (25.4) \\
\hline
\\
\multicolumn{6}{p{5in}}{Notes: IPW = inverse-probability-weighted; 
P/NP = parametric/nonparametric nuisance estimation based on $\bX$ (Cor) or $\bX^*$ (Mis).}
\end{tabular} 
\end{center}
\end{table}
\end{comment}

In each setting, the proposed estimator performed as well or better than the plug-in and IPW versions. When the nuisance functions were estimated with correct parametric models, all methods gave small bias and RMSE, with the plug-in and proposed estimators slightly outperforming the IPW estimator in terms of RMSE. Under parametric misspecification, bias and RMSE were amplified for all estimators and the plug-in fared worst. A more interesting (but expected) story appeared with nonparametric nuisance estimation. There, the plug-in and IPW estimators show large bias and RMSE, since they are not expected to converge at $\sqrt{n}$ rates; in contrast, the proposed efficient estimator essentially matches its behavior when constructed based on correct parametric models (with only a slight loss in RMSE). This is indicative of the fact that the proposed estimator only requires $n^{1/4}$ rates on nuisance estimation to achieve full efficiency and in general has second-order bias. This behavior appears to hold in our simulations even for nonparametric estimation using $\bX^*$, i.e., when the true model is not used directly.

We also assessed the uniform coverage of our proposed multiplier bootstrap confidence bands (as usual, we say a band covers if it contains the true curve entirely for all $\delta \in \mathcal{D}$). \textcolor{black}{Results are given in Table 1. As expected, coverage is very poor when nuisance functions are estimated with misspecified parametric models.} Coverage was near the nominal level (95\%) in large samples as long as nuisance functions were estimated with correct parametric models or nonparametrically using the non-transformed covariates $\bX$ (coverage was slightly diminished  for nonparametric nuisance estimation based on the misspecified $\bX^*$). 

\begin{table}[h!]
\caption{Coverage of proposed uniform 95\% confidence band across 500 simulations. \label{tab:simtab2}}
\begin{center}
\begin{tabular}{lrrrr}
\hline
Sample & \multicolumn{4}{c}{Coverage (\%) for setting:} \\
size $n$ & Cor P & Mis P & Cor NP & Mis NP  \\
\hline
%100 & 90.6 & 84.0 & 93.0 & 88.6 \\
500 & 92.4 & 77.0 & 93.0 & 88.0 \\
1000 & 95.2 & 67.6 & 95.6 & 92.4 \\
5000 & 94.8 & 12.4 & 94.2 & 89.4 \\
\hline
\end{tabular} 
\end{center}
\end{table}

\subsection{Application}
\label{sec:app}

Here we illustrate the use of incremental intervention effects with a reanalysis of the National Longitudinal Survey of Youth 1997 data used by \textcite{huebner2005effect}, \textcite{bacak2015marginal}, and others to study the effects of incarceration on marriage. Incarceration is a colossal industry in the United States, with over 2.3 million people currently confined in a correctional facility and at least twice that number held on probation or parole \autocite{wagner2016mass}. There is a large literature on unintended effects of this mass incarceration, with numerous studies pointing to negative impacts on various aspects of employment, health, social ties, psychology, and more \autocite{pattillo2004imprisoning, clear2009imprisoning}. Effects of incarceration on marriage are important since marriage is expected to yield, for example, better family and social support, better outcomes for children, and less recidivism, among other benefits \autocite{huebner2005effect, clear2009imprisoning}. \textcite{bacak2015marginal} were the first to study this question while specifically accounting for time-varying confounders, such as employment and earnings, and we refer there for more motivation and background.

The National Longitudinal Survey of Youth 1997 data consists of yearly measures across 14 timepoints, from 1997 to 2010, for participants who were 12--16 years old at the initial survey. The data include demographic information (e.g., age, race, gender, parent's education), various delinquency indicators (e.g., age at first sex, measures of drug use and gang membership, delinquency scores), as well as numerous time-varying measures (e.g., employment, earnings, marriage and incarceration history). Following \textcite{bacak2015marginal}, we use the final 10 timepoints from 2001--2010, restrict the analysis to the 4781 individuals with a non-zero delinquency score at baseline, and use as outcome $Y$ the indicator of marriage at the end of the study (i.e., in 2010).

\textcite{bacak2015marginal} used a standard marginal structural model approach to study effects of static incarceration trajectories, which has some limitations. First, it requires a parametric model to describe how incarceration trajectories affect marriage rates. In particular \textcite{bacak2015marginal} used $\E(Y^{\overline{a}_T}) = \expit( \beta_0 + \beta_1 \sum_t a_t)$, which only allows marriage prevalence to depend on total time spent incarcerated. This kind of assumption is very common in practice but is quite restrictive, especially since a saturated structural model in this case would have $2^{10}=1024$ parameters, instead of only two. Hence the data only inform $2/1024=0.2\%$ of the possible parameter values. \textcolor{black}{In fact if the model is slightly elaborated, e.g., to $\E(Y^{\overline{a}_T}) = \expit( \beta_0 + \sum_t \beta_{1t} a_t)$ so that $\beta_1$ can vary with time, then a standard weighting  estimator fails and no coefficient estimates can be found.}  Another limitation is that \textcite{bacak2015marginal} used parametric inverse probability weighting to estimate $(\beta_0,\beta_1)$ (partly for pedagogic purposes), but this is both inefficient and likely biased due to propensity score model misspecification. Perhaps most importantly, a standard marginal structural model setup requires imagining sending \textit{all} or \textit{none} of the study participants to prison at each time. However, positivity is likely violated here since some individuals may be necessarily incarcerated at some times (e.g., due to multiple-year sentences) or have essentially zero chance of incarceration (based on demographic or other characteristics).  These limitations are not at all unique to the analysis of \textcite{bacak2015marginal}, but instead are common to many observational marginal structural model analyses; we build on their analysis by instead estimating incremental incarceration effects, which require neither any parametric models nor any positivity assumptions. 

Specifically we estimated the incremental effect curve $\psi(\delta)$, which in this setting represents the marriage prevalence at the end of the study if the odds of incarceration were multiplied by factor $\delta$. We used Random Forests (via the \verb|ranger| package in R) to estimate all nuisance functions $\pi_t$ and $m_t$ as described in Algorithm 1 (with $K=10$-fold sample splitting), and computed pointwise and uniform confidence bands as in Sections \ref{sec:clt} and  \ref{sec:test} (with 10,000 bootstrap replications). Results are shown in Figure~\ref{fig:app_plot}.

\begin{figure}[h!]
\begin{center}
\includegraphics[width=.64\textwidth]{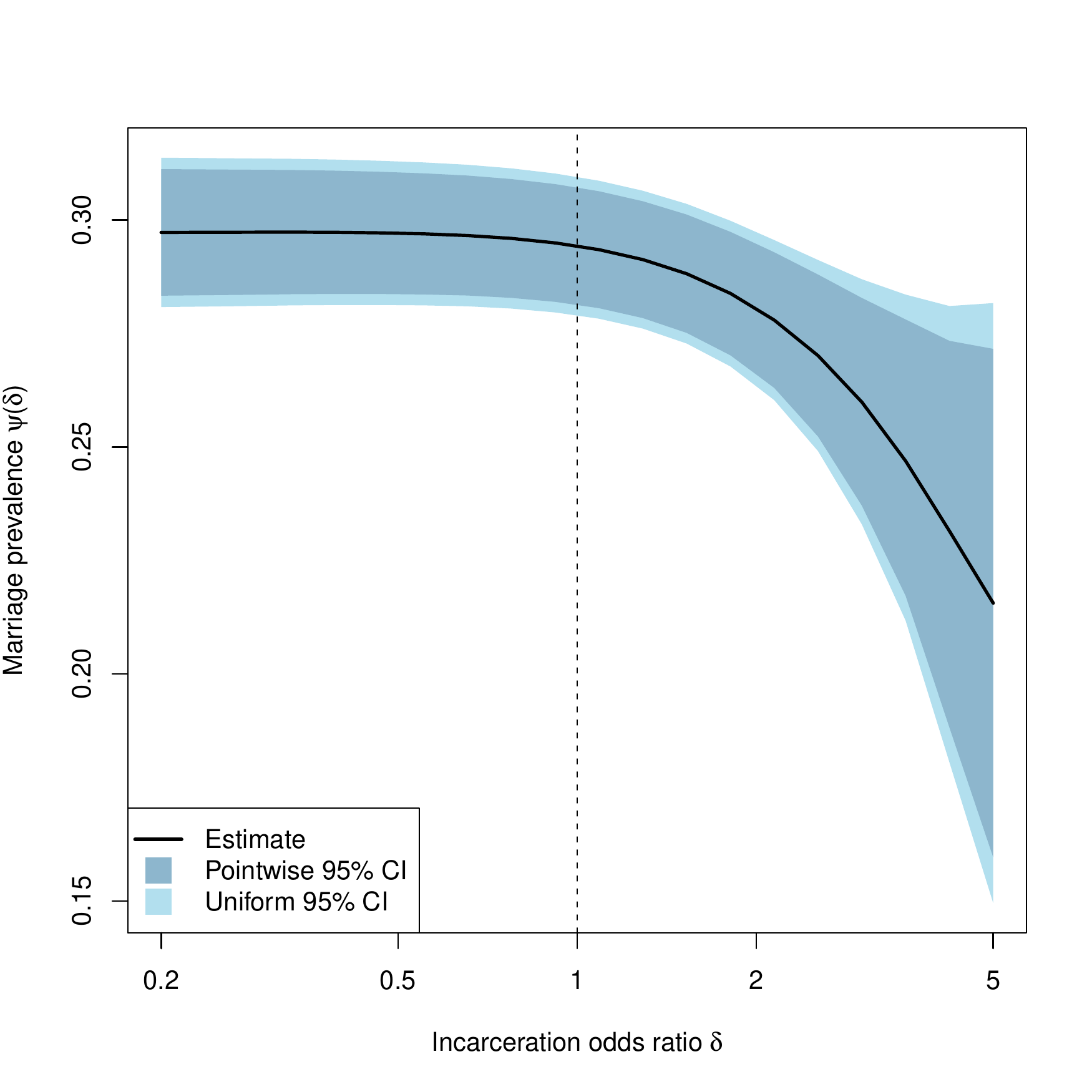}
\caption{Estimated marriage prevalence 10 years post-baseline, if the incarceration odds were multiplied by factor $\delta$, with pointwise and uniform 95\% confidence bands.} \label{fig:app_plot}
\end{center}
\end{figure}

We find strong evidence (assuming no unmeasured confounding and consistency) that incarceration negatively impacts marriage rates. First, we reject the null hypothesis of no incremental effect of incarceration on marriage ($p =0.049$) over the range $\delta \in [0.2,5]$. More specifically we estimate that, if incarceration odds were increased proportionally for all individuals, marriage prevalence would drop from $\Pn(Y)=29.4\%$ observationally to 28.1\% if the odds doubled (OR=0.94, 95\% CI: 0.87--1.00), and to 23.6\% if the odds were multiplied four-fold (OR=0.74, 95\% CI: 0.59--0.91). Conversely, we estimate that marriage prevalence would only increase to 29.7\% if the odds of incarceration were halved (OR=1.01, 95\% CI: 0.95--1.08); the prevalence and odds ratio are the same if the odds were quartered. These results suggest that marriage rates might be more affected by increased rather than decreased incarceration (i.e., the curve in Figure \ref{fig:app_plot} is nonlinear, with larger slope for $\delta>1$). This analysis provides considerably more nuance than a simple marginal structural model fit, and requires none of the parametric and positivity assumptions.

\section{Discussion}

In this paper we have proposed incremental intervention effects, which are based on shifting propensity scores rather than setting treatment values. We showed that these effects can be identified and estimated without any positivity or parametric assumptions, established general efficiency theory, and constructed influence-function-based estimators that yield fast rates of convergence even when based on flexible nonparametric regression tools. We also developed an approach for uniform inference and a new test of no treatment effect, and applied the methods in a longitudinal study of incarceration effects on marriage.  

There are a few caveats to our developments that are worth mentioning. First, we expect incremental intervention effects to play a more descriptive than prescriptive role compared to other approaches. Specifically, they give an interpretable picture of what would happen if exposure were increased or decreased in a natural way, but will likely be less useful for informing specific treatment decisions. For example in our analysis from Section \ref{sec:app} the goal was to better understand the overall societal effects of mass incarceration; in cases where the goal is to learn how to best assign treatment, methods for optimal treatment regime estimation will likely be more relevant. However, note that it is certainly possible to estimate the optimal incremental regime $q_t(\bh_t; \delta^*,\pi_t)$ for $\delta^* = \argmax_\delta \E(Y^{\bQ(\delta)})$; so in theory incremental effects could be used to construct specific treatment decision rules. 

Another caveat is that, in favor of computational efficiency, we have bypassed concerns about model compatibility when estimating the pseudo-regression functions $m_t$. It can be difficult to formulate models for all the $m_t$ functions that are compatible with each other, since $m_t$ has a complicated dependence on $m_{t+1}$ (as well as the propensity scores $\pi_{t+1}$ and covariate densities). To make these estimators fully compatible, we would need to model the conditional densities of the (high-dimensional) covariates and construct $\hat{m}_t$ based on the non-recursive expression in Theorem \ref{thm:eif}. However, we feel that if flexible enough estimators for $m_t$ are used, then model incompatibility will likely not be a major concern in practice, particularly relative to the computational benefits. This issue also arises in estimating standard longitudinal causal effects \autocite{scharfstein1999adjusting, murphy2001marginal}.

In future work we plan to pursue various extensions of incremental intervention effects. For example, it will be important to consider (i) interventions with increment parameters $\delta=\delta(\bh_t)$ that depend on time and past covariate history, (ii) estimation of how mean outcomes under different interventions vary with covariates (effect modification), (iii) extensions for settings with multivalued treatments and/or censored outcomes, and (iv) increment parameters based on risk ratios or other shifts,  rather than odds ratios.

\pagebreak

\stepcounter{section}
\printbibliography[title={\thesection \ \ \ References}]

\pagebreak

\spacingset{1}
\setcounter{page}{1}

\section{Appendix} 

\subsection{Proof of Theorem \ref{thm:ident}}

First we give a useful identification result for general stochastic intervention effects.

\begin{lemma}
Let $\bQ=(Q_1,...,Q_T)$ denote a general stochastic intervention in which treatment at time $t$ is randomly assigned according to distribution function $Q_t(a_t \mid \bh_t)$. Under Assumptions 1--2, and if (weak) positivity holds in the sense that
$$ d\Pb(a_t \mid \bh_t) = 0 \implies dQ_t(a_t \mid \bh_t) = 0 $$
then the mean outcome $\E(Y^\bQ)$ under the intervention is identified by
$$ \psi^*(\bQ)= \int_\mathcal{A} \int_\mathcal{X} \E(Y \mid \overline\bX_T=\overline\bx_T, \overline{A}_T=\overline{a}_T) \prod_{t=1}^T dQ_t(a_t \mid \bh_t)  \ d\Pb(\bx_{t} \mid \bh_{t-1},a_{t-1}) , $$
where $\mathcal{A}=\mathcal{A}_1 \times \dots \mathcal{A}_T$ and $\mathcal{X}=\mathcal{X}_1 \dots \mathcal{X}_T$.
\end{lemma}

\begin{proof}
This essentially follows by the g-formula of \textcite{robins1986new}. Let underbars denote the future of a sequence so that for example $\underline\bQ_t=(Q_t,...,Q_T)$. Then we have the recursion
\begin{align*}
\E(Y^{(\overline{a}_{t-1}, \underline\bQ_t)} &\mid \bH_{t-1}, A_{t-1}) = \E\{ \E(Y^{(\overline{a}_{t-1}, \underline\bQ_t)} \mid \bH_t , Q_t) \mid \bH_{t-1}, A_{t-1} \}  \\
&= \int \E(Y^{(\overline{a}_{t-1},q_t, \underline\bQ_{t+1})}  \mid \bH_t=\bh_t, Q_t=q_t ) \ dQ_t(q_t \mid \bh_t) \ d\Pb(\bx_t \mid \bh_{t-1}, a_{t-1}) \\
&= \int \E(Y^{(\overline{a}_{t-1},q_t, \underline\bQ_{t+1})}  \mid \bH_t=\bh_t, A_t=q_t ) \ dQ_t(q_t \mid \bh_t) \ d\Pb(\bx_t \mid \bh_{t-1}, a_{t-1}) \\
&= \int \E(Y^{(\overline{a}_t, \underline\bQ_{t+1})}  \mid \bH_t=\bh_t, A_t=a_t ) \ dQ_t(a_t \mid \bh_t) \ d\Pb(\bx_t \mid \bh_{t-1}, a_{t-1})
\end{align*}
for $t=1, ..., T$, where the first equality follows by iterated expectation, the second by definition, the third since $Q_t \ind Y^\bQ \mid \bH_t$ (by definition) along with exchangeability (Assumption 2), and the fourth by simply rewriting the index $q_t$ as $a_t$.  The weak positivity condition is required so that the above outer expectation is well-defined (that the inner expectation may not be is fine since, by positivity, in such cases the multiplier $dQ_t$ will be zero).

Therefore applying the above $T$ times yields
\begin{align*}
\E(Y^{\bQ}) &= \int_{\mathcal{A}_1}  \int_{\mathcal{X}_1} \E(Y^{(a_1, \underline\bQ_2)} \mid \bX_1=\bx_1, A_1=a_1) \ dQ_1(a_1 \mid \bx_1) \ d\Pb(\bx_1) \\
&= \int\limits_{\mathcal{A}_1 \times \mathcal{A}_2}   \int\limits_{\mathcal{X}_1 \times \mathcal{X}_2}  \E(Y^{(\overline{a}_2, \underline\bQ_3)} \mid \bH_2=\bh_2, A_2=a_2) \prod_{t=1}^2 dQ_t(a_t \mid \bh_t)  \ d\Pb(\bx_{t} \mid \bh_{t-1},a_{t-1}) \\
&= \int_\mathcal{A} \int_\mathcal{X} \E(Y^{\overline{a}_T} \mid \overline\bX_T=\overline\bx_T, \overline{A}_T=\overline{a}_T) \prod_{t=1}^T dQ_t(a_t \mid \bh_t)  \ d\Pb(\bx_{t} \mid \bh_{t-1},a_{t-1}) \\
&= \int_\mathcal{A} \int_\mathcal{X} \E(Y \mid \overline\bX_T=\overline\bx_T, \overline{A}_T=\overline{a}_T) \prod_{t=1}^T dQ_t(a_t \mid \bh_t)  \ d\Pb(\bx_{t} \mid \bh_{t-1},a_{t-1})
\end{align*}
where the last equality follows by consistency (Assumption 1).
\end{proof}

Now Theorem \ref{thm:ident} follows from Lemma 1, letting
$$ dQ_t(a_t \mid \bh_t) = \frac{ a_t \delta \pi_t(\bh_t) + (1-a_t) \{1-\pi_t(\bh_t)\} }{ \delta \pi_t(\bh_t) + 1-\pi_t(\bh_t) } $$
and noting that
\begin{align*}
\pi_t(\bh_t) = 0 &\implies dQ_t(1 \mid \bh_t) = 0 \\
\pi_t(\bh_t)=1 &\implies dQ_t(0 \mid \bh_t) = 0 
\end{align*}
so that the weak positivity condition is automatically satisfied by our choice of $dQ_t$.

\subsection{Proof of Theorem \ref{thm:eif}}
\label{sec:proof_eif}

First we derive the efficient influence function for a general stochastic intervention effect when the intervention distribution $Q$ does not depend on the observed data distribution $\Pb$. 

\begin{lemma}
\label{lem:eif_qknown}
Suppose $\bQ$ is a known stochastic intervention not depending on $\Pb$. Define
$$ m_t(\bh_t,a_t) = \int_{\mathcal{R}_t} \mu(\bh_T,a_T) \prod_{s={t+1}}^T  dQ_s(a_s \mid \bh_s) \ d\Pb(\bx_{s} \mid \bh_{s-1},a_{s-1}) $$
for $t=0,...,T-1$ and $\mathcal{R}_t= (\mathcal{H}_T \times \mathcal{A}_T) \setminus \mathcal{H}_t$, and let $m_T(\bh_T,a_T)=\mu(\bh_T,a_T)$ and $m_{T+1}(\bh_{T+1},a_{T+1})=Y$. Then the efficient influence function for $\psi^*(\bQ)=m_0$ is
\begin{align*}
%&\sum_{t=1}^{T+1} \left\{ \int_{\mathcal{A}_t} m_t(\bH_t,a_t) \ dQ_t(a_t \mid \bH_t) - m_{t-1}(\bH_{t-1},A_{t-1}) \right\} \prod_{s=0}^{t-1} \frac{ dQ_s(A_s \mid \bH_s) }{ d\Pb(A_s \mid \bH_s) } \\
& \sum_{t=0}^{T} \left\{ \int_{\mathcal{A}_{t+1}} \!\!\! m_{t+1}(\bH_{t+1},a_{t+1}) \ dQ_{t+1}(a_{t+1} \mid \bH_{t+1}) - m_{t}(\bH_{t},A_{t}) \right\} \prod_{s=0}^{t} \frac{ dQ_s(A_s \mid \bH_s) }{ d\Pb(A_s \mid \bH_s) } \\
&= \sum_{t=1}^T \left\{ \int_{\mathcal{A}_t} m_t(\bH_t,a_t) \ dQ_t(a_t \mid \bH_t) -  m_t ( \bH_t,A_{t}) \frac{ dQ_t(A_t \mid \bH_t) }{ d\Pb(A_t \mid \bH_t) } \right\}  \prod_{s=0}^{t-1} \frac{ dQ_s(A_s \mid \bH_s) }{ d\Pb(A_s \mid \bH_s) } \\
& \hspace{.5in} + \prod_{s=1}^T \frac{ dQ_s(A_s \mid \bH_s) }{ d\Pb(A_s \mid \bH_s) } Y - \psi^*(\bQ) 
\end{align*}
where we define $dQ_{T+1}=1$ and $dQ_0(a_0 \mid \bh_0)/d\Pb(a_0 \mid \bh_0)= 1$.
\end{lemma}

\begin{lemma}
\label{lem:contribution}
Suppose $\bQ$ depends on $\Pb$, and let $\{\one(\bH_t=\bh_t)/d\Pb(\bh_t) \} \phi_t(\bH_t, A_t; a_t)$ 
%$$ \frac{\one(\bH_t=\bh_t)}{d\Pb(\bh_t)} \phi_t(\bH_t, A_t; Q_t) $$
denote the efficient influence function for $dQ_t(a_t \mid \bh_t)$. Then the efficient influence for $\E(Y^\bQ)$ allowing $\bQ$ to depend on $\Pb$ is given by
$$ \varphi^*(\bQ) + \sum_{t=1}^T \left\{ \prod_{s=0}^{t-1}  \frac{ dQ_s(A_s \mid \bH_s) }{ d\Pb(A_s \mid \bH_s) } \right\} \int_{\mathcal{A}_t} \phi_t(\bH_t,A_t; a_t) m_t( \bH_t, a_t) \ d\nu(a_t) $$
where $\varphi^*(\bQ)$ denotes the efficient influence function from Lemma 2 under an intervention $\bQ$ not depending on $\Pb$, and $\nu$ is a dominating measure for the distribution of $A_t$.
\end{lemma}
 
The proofs of Lemmas \ref{lem:eif_qknown} and \ref{lem:contribution} are based on chain rule arguments stemming from the fact that the efficient influence function is a pathwise derivative. In particular (in a nonparametric model) the efficient influence function for parameter $\psi=\psi(\Pb)$ is the function $\varphi(\Pb)$ satisfying
$$ \frac{\partial}{\partial \epsilon} \psi(\Pb_\epsilon) \Bigm|_{\epsilon=0} = \int \varphi(\Pb) \left( \frac{\partial}{\partial \epsilon} \log d\Pb_\epsilon \right) \Bigm|_{\epsilon=0} d\Pb $$
where $\{\Pb_\epsilon : \epsilon \in \R\} $ is a smooth parametric submodel with $\Pb_{\epsilon=0}=\Pb$. We omit the proofs since they are lengthy and not particularly illuminating; however we plan to include them in a forthcoming paper on general stochastic interventions.

\begin{lemma}
\label{lem:eif_contribution}
The efficient influence function for 
$$dQ_t(a_t \mid \bh_t) = \frac{a_t \delta \pi_t(\bh_t) + (1-a_t) \{1-\pi(\bh_t)\}}{\delta \pi_t(\bh_t) + 1-\pi_t(\bh_t)} $$
is given by $\{\one(\bH_t=\bh_t)/d\Pb(\bh_t) \} \phi_t(\bH_t, A_t; a_t)$ where $\phi_t(\bH_t, A_t; a_t)$ equals
$$ \frac{ (2 a_t-1) \delta \{A_t - \pi_t(\bH_t)\}}{ \{ \delta \pi_t(\bH_t) + 1-\pi_t(\bH_t)\}^2}  . $$
\end{lemma}

\begin{proof}
This result also follows from the chain rule, together with the fact that the efficient influence function for $\pi_t$ is given by
$$ \one(\bH_t=\bh_t) \{A_t - \pi_t(\bh_t)\} / d\Pb(\bh_t) . $$ 
\end{proof}

\subsection{Z-Estimator Algorithm}
\label{sec:algorithm}

\begin{algorithm}[Z-estimator algorithm] For each $\delta$:
\begin{enumerate} 
\item Regress $A_t$ on $\bH_t$, obtain predicted values $\hat\pi_t(\bH_t)$ for each subject/time.
\item Construct time-dependent weights $W_t = \frac{\delta A_t + 1-A_t}{\delta \hat\pi_t(\bH_t) + 1-\hat\pi_t(\bH_t)}$ for each subject/time.
\item Calculate cumulative product weight $\widetilde{W}_t = \prod_{s=1}^t W_s$ for each subject/time.
\item For each time $t=T, T-1, ..., 1$ (starting with $R_{T+1}=Y$):
\begin{enumerate}
\item Regress $R_{t+1}$ on $(\bH_t,A_t)$, obtain predicted values $\hat{m}_t(\bH_t,1)$ and $\hat{m}_t(\bH_t,0)$.
\item Construct pseudo-outcome $R_t=\frac{ \delta \hat\pi_t(\bH_t) \hat{m}_t(\bH_t,1) + \{1- \hat\pi_t(\bH_t)\} \hat{m}_t(\bH_t,0) } { \delta \hat\pi_t(\bH_t) + 1 - \hat\pi_t(\bH_t) }$.
\end{enumerate}
\item For each subject compute $\varphi = \widetilde{W}_T Y + \sum_t \widetilde{W}_t V_t R_t $ where $V_t = \frac{A_t \{1- \hat\pi_t(\bH_t)\} - (1-A_t) \delta \hat\pi_t(\bH_t) }{ \delta/(1-\delta)}$.
\item Set $\hat\psi^*(\delta)$ to be the average of the $\varphi$ values across subjects.
\end{enumerate}
\end{algorithm}

\subsection{Proof of Theorem \ref{thm:unif_clt}}
\label{proof:unif_clt}

Let $\| f \|_\mathcal{D} = \sup_{\delta \in \mathcal{D}} | f(\delta) |$ denote the supremum norm over $\mathcal{D}$, and define the processes
\begin{align*}
\widehat\Psi_n(\delta) &= \sqrt{n} \{\hat\psi(\delta) - \psi(\delta)\} / \hat\sigma(\delta) \\
\widetilde\Psi_n(\delta) &= \sqrt{n} \{\hat\psi(\delta) - \psi(\delta)\} / \sigma(\delta) \\
\Psi_n(\delta) &= \Gn[ \{ \varphi(\bZ;\boldsymbol{\eta},\delta)-\psi(\delta)\} /\sigma(\delta) ]  = \Gn\{ \widetilde\varphi(\bZ;\boldsymbol\eta,\delta)\}
\end{align*}
where $\Gn=\sqrt{n}(\Pn-\Pb)$ is the empirical process on the full sample process as usual. Also let $\Gb(\delta)$ denote the mean-zero Gaussian process with covariance $\E\{ \widetilde\varphi(\bZ;\boldsymbol\eta,\delta_1)\widetilde\varphi(\bZ;\boldsymbol\eta,\delta_2) \}$ as in the main text. \\

In this proof we will show that
$$ \Psi_n(\cdot) \indist \Gb(\cdot) \text{ in } \ell^\infty(\mathcal{D})  \ \text{ and } \ \| \widehat\Psi_n - \Psi_n \|_\mathcal{D} = o_\Pb(1)  $$
%$$ \| \widehat\Psi_n - \Gb \|_\mathcal{D} \leq \| \widehat\Psi_n - \Psi_n \|_\mathcal{D} + \| \Psi_n - \Gb \|_\mathcal{D} = o_\Pb(1) $$
which yields the desired result. The first statement will be true if the influence function $\widetilde\varphi$ is a smooth enough function of $\delta$, and the second if the nuisance estimators $\hat{\boldsymbol\eta}$ are consistent and converging at a sufficiently fast rate. \\

The first statement follows since the function class $\mathcal{F}_{\boldsymbol{\overline\eta}} = \{\varphi(\cdot; \boldsymbol{\overline\eta}, \delta) : \delta \in \mathcal{D}\}$ is Lipschitz and thus has a finite bracketing integral for any fixed $\boldsymbol{\overline\eta}$. Recall the $L_2(\Pb)$ bracketing integral of class $\mathcal{F}$ with envelope $F$ is given by
$$ J_{[ \ \! ]}(\mathcal{F}) = \int_0^1 \sqrt{ 1+ \log N_{[ \ \! ]}(\epsilon \| F\| , \mathcal{F}, L_2(\Pb))} \ d\epsilon  $$
where $N_{[ \ \! ]}(\epsilon , \mathcal{F}, L_2(\Pb))$ is the $L_2(\Pb)$ bracketing number, i.e., the minimum number of $\epsilon$ brackets in $L_2(\Pb)$ needed to cover the class $\mathcal{F}$ with envelope function $F$. That  $\mathcal{F}_{\boldsymbol{\overline\eta}}$ is Lipschitz (and thus the bracketing integral is finite) follows from the fact that $\varphi$ is a sum of products of Lipschitz functions and $\mathcal{D}$ is bounded.  We show this by showing that the corresponding derivatives are all bounded, specifically
\begin{align*}
\left| \frac{\partial}{\partial \delta} \left[ \frac{a_t \{1-\pi_t(\bh_t)\} - (1-a_t) \delta \pi_t(\bh_t) }{ \delta/(1-\delta)} \right] \right| &= \left| \frac{a_t\{1-\pi_t(\bh_t)\} }{\delta^2} - (1-a_t) \pi_t(\bh_t) \right| \leq 1 + 1/\delta_\ell^2 \\
 \left| \frac{\partial}{\partial \delta} \left[ \frac{ \delta \pi_t(\bh_t) m_t(\bh_t,1) + \{1-\pi_t(\bh_t)\} m_t(\bh_t,0) } { \delta \pi_t(\bh_t) + 1 - \pi_t(\bh_t) } \right] \right| &= \left| \frac{\pi_t(\bh_t) \{1-\pi_t(\bh_t)\} \{m_t(\bh_t,1)-m_t(\bh_t,0)\} }{\{ \delta \pi_t(\bh_t) + 1-\pi_t(\bh_t)\}^2} \right| \\
 &\leq  | m_t(\bh_t, 1) - m_t(\bh_t,0) | / \delta_\ell^2 \\
  \left| \frac{\partial}{\partial \delta} \left[ \frac{ \delta a_t + 1-a_t }{ \delta \pi_t(\bh_t) + 1 - \pi_t(\bh_t) } \right] \right| &= \left| \frac{a_t - \pi_t(\bh_t)}{\{ \delta \pi_t(\bh_t) + 1-\pi_t(\bh_t)\}^2} \right| \leq 1/\delta_\ell^2
\end{align*}
where we used the fact that, for all $0 \leq \pi_t(\bh_t) \leq 1$, we have
$$ \{ \delta \pi_t(\bh_t) + 1-\pi_t(\bh_t)\} \in [ \delta \wedge 1 , \delta \vee 1 ] \subseteq [\delta_\ell, \delta_u] . $$
Therefore $\Psi_n(\cdot) \indist \Gb(\cdot)$ since a function class with finite bracketing integral is necessarily Donsker (e.g., Theorem 2.5.6 in \textcite{van1996weak}).

Now we consider the second statement, that $\| \widehat\Psi_n - \Psi_n \|_\mathcal{D} =o_\Pb(1)$.  First note that
\begin{align*}
\| \widehat\Psi_n - \Psi_n \|_\mathcal{D} &= \| (\widetilde\Psi_n - \Psi_n) (\sigma/\hat\sigma) + \Psi_n (\sigma-\hat\sigma) / \hat\sigma \|_\mathcal{D} \\
&\leq \| \widetilde\Psi_n - \Psi_n \|_\mathcal{D} \| \sigma/\hat\sigma \|_\mathcal{D} + \| \sigma/\hat\sigma - 1\|_\mathcal{D} \|\Psi_n \|_\mathcal{D} \\
& \lesssim \| \widetilde\Psi_n - \Psi_n \|_\mathcal{D}  + o_\Pb(1)
\end{align*}
where the last inequality follows since $\|\hat\sigma/\sigma - 1\|_\mathcal{D}=o_\Pb(1)$ by Assumption 3 of Theorem \ref{thm:unif_clt}, and $\|\Psi_n \|_\mathcal{D}=O_\Pb(1)$ follows from, e.g., Theorem 2.14.2 in \textcite{van1996weak}, since the function class $\mathcal{F}_{\boldsymbol\eta}$ has finite bracketing integral as shown above. \\

Now let $N=n/K$ be the sample size in any group $k=1,...,K$, and denote the empirical process over group $k$ units by $\Gn^k=\sqrt{N}(\Pn^k - \Pb)$. Then we have 
\begin{align*}
\widetilde\Psi_n(\delta) &- \Psi_n(\delta) = \frac{ \hat\psi(\delta) - \psi(\delta) }{\sigma(\delta) / \sqrt{n} } -  \Gn \{ \widetilde\varphi(\bZ;\boldsymbol\eta, \delta) \} \\
&= \frac{\sqrt{n}}{\sigma(\delta)} \frac{1}{K} \sum_{k=1}^K \Big[ \Pn^k\{ \varphi(\bZ;\hat{\boldsymbol\eta}_{\text{-}k},\delta)\} - \psi(\delta) - (\Pn-\Pb)\varphi(\bZ;\boldsymbol\eta,\delta) \Big] \\
&= \frac{\sqrt{n}}{K \sigma(\delta)} \sum_{k=1}^K \Big[ \frac{1}{\sqrt{N}} \Gn^k \Big\{ \varphi(\bZ; \hat{\boldsymbol\eta}_{\text{-}k},\delta) -\varphi(\bZ;\boldsymbol\eta, \delta) \Big\} + \Pb\Big\{ \varphi(\bZ; \hat{\boldsymbol\eta}_{\text{-}k},\delta) -  \varphi(\bZ;\boldsymbol\eta,\delta) \Big\} \Big] \\
&\equiv B_{n,1}(\delta) + B_{n,2}(\delta)
 \end{align*}
where the first two equalities follow by definition, and the third by rearranging and noting that 
$ \psi(\delta) = \Pb\{\varphi(\bZ;\boldsymbol\eta,\delta)\} $
and 
$ \sum_k \Pn^k\{\varphi(\bZ;\boldsymbol\eta,\delta)\} = \sum_k \Pn\{\varphi(\bZ;\boldsymbol\eta,\delta)\} $. Now we will analyze the two pieces $B_{n,1}$ and $B_{n,2}$ in turn; showing that their supremum norms are both $o_\Pb(1)$ completes the proof. \\

For $B_{n,1}$, we have by the triangle inequality and since $K$ is fixed (independent of total sample size $n$), that
\begin{align*}
\| B_{n,1} \|_\mathcal{D} &= \sup_{\delta \in \mathcal{D}} \left| \frac{1}{\sqrt{K} \sigma(\delta)} \sum_{k=1}^K \Gn^k \Big\{ \varphi(\bZ; \hat{\boldsymbol\eta}_{\text{-}k},\delta) -\varphi(\bZ;\boldsymbol\eta, \delta) \Big\} \right| \\
&\lesssim \max_k \sup_{f \in \mathcal{F}_n^k} \left| \Gn(f) \right| , 
\end{align*}
where $\mathcal{F}_n^k = \mathcal{F}_{\hat{\boldsymbol\eta}_{\text{-}k}} - \mathcal{F}_{\boldsymbol{\eta}}$ for the function class $\mathcal{F}_{\boldsymbol\eta} = \{ \varphi(\cdot; \boldsymbol\eta, \delta) : \delta \in \mathcal{D} \}$ from before. Viewing the nuisance functions $\hat{\boldsymbol\eta}_{\text{-}k}$ as fixed given the training data $\bD_0^k = \{\bZ_i : S_i \neq k\}$, we can apply Theorem 2.14.2 in \textcite{van1996weak} to obtain
$$ \E \left\{ \sup_{f \in \mathcal{F}_n^k} \left| \Gn(f) \right|  \Bigm| \bD_0^k \right\} \lesssim \| F_n^k \| \int_0^1 \sqrt{ 1 + \log N_{[ \ \! ]}(\epsilon  \| F_n^k \|, \mathcal{F}_n^k , L_2(\Pb) )} \ d\epsilon   $$
for envelope $F_n^k$. If we take $F_n^k(\bz)=\sup_{\delta \in \mathcal{D}} | \varphi(\bz; \hat{\boldsymbol\eta}_{\text{-}k},\delta) - \varphi(\bz; \boldsymbol\eta,\delta)| $ then the first term $\| F_n^k \|$ in the product above is $o_\Pb(1)$. Although the bracketing integral is finite for any fixed $\boldsymbol\eta$, here the function class depends on $n$ through $\hat{\boldsymbol\eta}_{\text{-}k}$ so we need a more careful analysis.

Specifically, since $\mathcal{F}_n^k$ is Lipschitz, by Theorem 2.7.2 of \textcite{van1996weak} we have
$$ \log N_{[ \ \! ]}(\epsilon  \| F_n^k \|, \mathcal{F}_n^k , L_2(\Pb) ) \lesssim \frac{1}{\epsilon  \| F_n^k \|} . $$
Therefore, letting $C_n^k= \| F_n^k \|$,
\begin{align*}
\| F_n^k \|  \int_0^1 &\sqrt{ 1 + \log N_{[ \ \! ]}(\epsilon  \| F_n^k \|, \mathcal{F}_n^k , L_2(\Pb) )} \ d\epsilon \lesssim C_n^k \int_0^1 \sqrt{ 1 + \frac{1}{\epsilon C_n^k} } \ d\epsilon \\
&= C_n^k \sqrt{1+\frac{1}{C_n^k}} + \frac{1}{2C_n^k} \log\left\{ 1 + 2C_n^k \left( 1+ \sqrt{1+\frac{1}{C_n^k}} \right) \right\} \\
&= \sqrt{ C_n^k(C_n^k + 1)} + (1/2) \log\left\{ 1 + 2C_n^k \left( 1+ \sqrt{1+\frac{1}{C_n^k}} \right) \right\} 
\end{align*}
which tends to zero as $C_n^k \rightarrow 0$. Hence $\sup_{f \in \mathcal{F}_n^k} \left| \Gn(f) \right| = o_\Pb(1)$ for each $k$, and since there are only finitely many splits $K$, we have
$$ \| B_{n,1} \|_\mathcal{D} = o_\Pb(1) . $$

To analyze $B_{n,2}(\delta)$ we require some new notation, and at first we typically suppress any dependence on $\delta$ for simplicity. Let $\psi(\Pb;Q)$ denote the mean outcome under intervention $Q$ for a population corresponding to observed data distribution $\Pb$, and let $\varphi^*(\bz;\boldsymbol{\eta})$ denote its (centered) efficient influence function when $Q$ does not depend on $\Pb$, as given in Lemma \ref{lem:eif_qknown}, which depends on nuisance functions $\boldsymbol{\eta} = (\mathbf{m},\boldsymbol\pi) = (m_0, m_1, ..., m_T, \pi_1, \pi_2, ..., \pi_T)$. Similarly let $\zeta(\bz;\boldsymbol{\eta})$ denote the contribution to the efficient influence function $\varphi^*(\bz;\boldsymbol\eta)$ due to estimating $Q$ when it depends on $\Pb$, as given in Lemma \ref{lem:contribution}. Then by definition
$$ \varphi(\bz; \boldsymbol\eta,\delta) = \varphi^*(\bz;\boldsymbol{\eta}) + \psi(\Pb;Q) + \zeta(\bz;\boldsymbol{\eta}) . $$
Hence, for any $\boldsymbol{\overline\eta}$ we can write $(1/\sqrt{n}) B_{n,2}(\delta)$ as
\begin{align*}
\Pb\Big\{ \varphi(\bZ; \boldsymbol{\overline\eta},\delta) &-  \varphi(\bZ;\boldsymbol\eta,\delta) \Big\} = \int \{ \varphi^*(\bz; \boldsymbol{\overline\eta}) + \zeta(\bz;\boldsymbol{\overline\eta}) + \psi(\overline\Pb,\overline{Q}) \} \ d\Pb(\bz)  - \psi(\Pb,Q) \\
&= \int \varphi^*(\bz; \boldsymbol{\overline\eta}) \ d\Pb(\bz) + \psi(\overline\Pb;\overline{Q}) - \psi(\Pb; \overline{Q}) \\
& \hspace{.4in} + \int \zeta(\bz;\boldsymbol{\overline\eta}) \ d\Pb(\bz) + \psi(\Pb;\overline{Q}) - \psi(\Pb;Q)
\end{align*}
where the first equality follows by definition and the second by rearranging. \\

In the following lemmas we analyze these two components of the remainder term $B_{n,2}(\delta)$. Our results keep the intervention distribution $Q$ completely general, and so can be applied to study other stochastic interventions, beyond those we focus on in this paper of the incremental propensity score variety.

\begin{lemma}
\label{lem:rem_qknown}
Let $\psi(\Pb;Q)$ denote the mean outcome under intervention $Q$ for a population corresponding to observed data distribution $\Pb$, and let $\varphi^*(\bz;\boldsymbol{\eta})$ denote its efficient influence function when $Q$ does not depend on $\Pb$, as given in Lemma \ref{lem:eif_qknown}, which depends on nuisance functions $\boldsymbol{\eta} = (\mathbf{m},\boldsymbol\pi) = (m_0, m_1, ..., m_T, \pi_1, \pi_2, ..., \pi_T)$. Then for two distributions $\Pb$ and $\overline\Pb$ (the latter with corresponding nuisance functions $\boldsymbol{\overline\eta}$) we have the expansion
\begin{align*}
\psi(\overline\Pb;Q) - \psi(\Pb;Q) &+ \int \varphi^*(\bz; \boldsymbol{\overline\eta}) \ d\Pb(\bz) \\
 &= \sum_{t=1}^T \sum_{s =1}^t \int \left( m_t^* - \overline{m}_t \right)\left( \frac{d\pi_s-d\overline\pi_s}{d\overline\pi_s} \right) \left( \prod_{r=1}^{s-1}   \frac{d\pi_r}{d\overline\pi_r} \right)  \left( \prod_{r=1}^t dQ_r \ d\Pb_r \right) 
%&= \sum_{t=1}^T \sum_{s=1}^t \E \{m_t^*(\bH_t,A_t) - \overline{m}_t(\bH_t,A_t)\}  \left( \frac{d\pi_s-d\overline\pi_s}{d\overline\pi_s} \right) \left( \prod_{r=1}^{s-1}   \frac{d\pi_r}{d\overline\pi_r} \right)
\end{align*}
where we define
$$ \overline{m}_t = \overline{m}_t(\bH_t,A_t) = \int \overline{m}_{t+1} \ dQ_{t+1} \ d\overline\Pb_{t+1} \ \ , \ \ m_t^* = \int \overline{m}_{t+1} \ dQ_{t+1} \ d\Pb_{t+1} , $$
$$ dQ_t = dQ_t(A_t \mid \bH_t) \ , \  d\pi_t = d\Pb(A_t \mid \bH_t) \ , \ d\Pb_t = d\Pb(\bX_{t} \mid \bH_{t-1},A_{t-1}) . $$
\end{lemma}

\begin{proof}
First note that
\begin{align*}
\E\{\varphi^*(\bZ;\boldsymbol{\overline\eta})\} &= \E \sum_{t=0}^T \left( \int \overline{m}_{t+1} \ dQ_{t+1} - \overline{m}_t \right) \prod_{s=0}^t \left( \frac{dQ_s}{d\overline\pi_s} \right) \\
&= \E \sum_{t=0}^T  \left( \int \int \overline{m}_{t+1} \ dQ_{t+1} \ d\Pb_{t+1} - \overline{m}_t \right)  \prod_{s=0}^t  \left( \frac{dQ_s}{d\overline\pi_s} \right) \\
&= \E \sum_{t=0}^T  \left( m_t^* - \overline{m}_t \right)  \prod_{s=0}^t  \left( \frac{dQ_s}{d\overline\pi_s} \right) \\
&= \sum_{t=0}^T \int \left( m_t^* - \overline{m}_t \right)  \prod_{s=0}^t  \left( \frac{dQ_s}{d\overline\pi_s} \right) d\pi_s \ d\Pb_s
\end{align*}
where the first equality follows by definition, the second by iterated expectation (conditioning on $(\bH_t,A_t)$ and averaging over $\bX_{t+1}$), the third by definition of $m_t^*$, and the fourth by repeated iterated expectation. Now we have
\begin{align*}
\sum_{t=0}^T \int &\left( m_t^* - \overline{m}_t \right)  \prod_{s=0}^t  \left( \frac{dQ_s}{d\overline\pi_s} \right) d\pi_s \ d\Pb_s \\
&= \sum_{t=1}^T \int \left( m_t^* - \overline{m}_t \right)  \left( \frac{d\pi_t-d\overline\pi_t}{d\overline\pi_t} \right) dQ_t \ d\Pb_t  \prod_{s=0}^{t-1}  \left( \frac{d\pi_s}{d\overline\pi_s} \right) dQ_s  \ d\Pb_s \\
& \hspace{.4in} + \sum_{t=1}^T \int \left( m_t^* - \overline{m}_t \right) dQ_t \ d\Pb_t  \prod_{s=0}^{t-1}  \left( \frac{d\pi_s}{d\overline\pi_s} \right) dQ_s  \ d\Pb_s + (m_0^* - \overline{m}_0) \\
%&= \sum_{t=1}^T \int \left( m_t^* - \overline{m}_t \right)  \left( \frac{d\pi_t-d\overline\pi_t}{d\overline\pi_t} \right) dQ_t \ d\Pb_t  \prod_{s=0}^{t-1}  \left( \frac{d\pi_s}{d\overline\pi_s} \right) dQ_s  \ d\Pb_s \\*
%& \hspace{.4in} + \sum_{t=1}^T \int \left( m_t^* - \overline{m}_t \right) dQ_t \ d\Pb_t  \left( \frac{d\pi_{t-1}-d\overline\pi_{t-1}}{d\overline\pi_{t-1}} \right) dQ_{t-1} \ d\Pb_{t-1} \prod_{s=0}^{t-2}  \left( \frac{d\pi_s}{d\overline\pi_s} \right) dQ_s  \ d\Pb_s \\*
%& \hspace{.4in} + \sum_{t=1}^T \int \left( m_t^* - \overline{m}_t \right) dQ_t \ d\Pb_t \ dQ_{t-1} \ d\Pb_{t-1} \prod_{s=0}^{t-2}  \left( \frac{d\pi_s}{d\overline\pi_s} \right) dQ_s  \ d\Pb_s + (m_0^* - \overline{m}_0)  \\
&= \sum_{t=1}^T \sum_{s=1}^{t} \int \left( m_t^* - \overline{m}_t \right) \left( \prod_{r=s}^t dQ_r \ d\Pb_r \right) \left( \frac{d\pi_s-d\overline\pi_s}{d\overline\pi_s} \right)  \prod_{r=1}^{s-1}  \left( \frac{d\pi_r}{d\overline\pi_r} \right) dQ_r  \ d\Pb_r \\
& \hspace{.4in} +  \sum_{t=1}^T \int \left( m_t^* - \overline{m}_t \right)  \prod_{s=1}^{t} dQ_s  \ d\Pb_s + (m_0^* - \overline{m}_0)
\end{align*}
where the first equality follows by adding and subtracting the second term in the sum (and separating the $t=0$ term), 
%the second follows by doing the same with the third term in the sum (where we use the convention that quantities at negative times like $dQ_{-1}$ are set to one),
and the second follows by repeating this process $t$ times (where we use the convention that quantities at negative times like $dQ_{-1}$ are set to one). The last terms in the last line above are a telescoping sum since
\begin{align*}
\sum_{t=1}^T \int \left( m_t^* - \overline{m}_t \right)  \prod_{s=1}^{t} dQ_s  \ d\Pb_s &= \sum_{t=1}^T \left( \int  m_t^* \prod_{s=1}^{t} dQ_s  \ d\Pb_s - \int \overline{m}_t \ dQ_t \ d\Pb_t \prod_{s=0}^{t-1} dQ_s \ d\Pb_s \right)  \\
&= \sum_{t=1}^T \left( \int m_t^* \prod_{s=1}^{t} dQ_s  \ d\Pb_s - \int m_{t-1}^* \prod_{s=0}^{t-1} dQ_s \ d\Pb_s \right) \\
&= \sum_{t=1}^T \int m_t^* \prod_{s=1}^{t} dQ_s  \ d\Pb_s - \sum_{t=1}^{T-1} \int m_t^* \prod_{s=1}^{t} dQ_s \ d\Pb_s  - m_0^* \\
&= \int m_T^* \prod_{s=1}^{T} dQ_s  \ d\Pb_s - m_0^* = m_0 - m_0^* .
\end{align*}
Therefore the result follows after rearranging and noting $\psi_Q^*(\Pb)=m_0$ and $\psi_Q^*(\overline\Pb)=\overline{m}_0$.
\end{proof}

\begin{lemma}
\label{lem:rem_contribution} 
Using the same notation as in Lemma \ref{lem:rem_qknown}, let $\zeta(\bZ;\boldsymbol{\eta})$ denote the contribution to the efficient influence function $\varphi^*(\bZ;\boldsymbol\eta)$ as given in Lemma \ref{lem:contribution}. Then for two intervention distributions $Q$ and $\overline{Q}$ (assumed to have densities $dQ_t$ and $d\overline{Q}_t$, respectively, for $t=1,...,T$, with respect to some dominating measure)  we have the expansion
\begin{align*}
\psi(\Pb;\overline{Q}) &- \psi(\Pb;Q) + \int  \zeta(\bz;\boldsymbol{\overline\eta}) \ d\Pb(\bz) \\
&= \sum_{t=1}^T \int (\overline\phi_t \ d\pi_t) (\overline{m}_t - m_t) \ d\nu \ d\Pb_t \left( \prod_{s=0}^{t-1}  \frac{ d\overline{Q}_s }{ d\overline\pi_s } d\pi_s \ d\Pb_s \right) \\
& \hspace{.4in} + \sum_{t=1}^T \sum_{s=1}^t \int (\overline\phi_t \ d\pi_t) \left( \frac{d\pi_s-d\overline\pi_s}{d\overline\pi_s} \right)  m_t \ d\nu \ d\Pb_t \left( \prod_{r=0}^{t-1} d\overline{Q}_r \ d\Pb_r \right) \left( \prod_{r=0}^{s-1}  \frac{ d\pi_s }{ d\overline\pi_s } \right) \\
& \hspace{.4in} + \sum_{t=1}^T \int  m_t (d\overline{Q}_t - dQ_t + \overline\phi_t \ d\pi_t) \ d\nu \ d\Pb_t \left( \prod_{s=0}^{t-1} dQ_s  \ d\Pb_s \right)
\end{align*}
\end{lemma}

\begin{proof}
First note that
\begin{align*}
\psi(\Pb;\overline{Q}) - \psi(\Pb;Q)  &= \int m_T \left( \prod_{t=t}^{T} d\overline{Q}_t  \ d\Pb_t - \prod_{t=t}^{T} d{Q}_t  \ d\Pb_t \right) \\
&= \int m_T (d\overline{Q}_T - dQ_T) \ d\Pb_T \prod_{t=1}^{T-1} d\overline{Q}_t  \ d\Pb_t + \int m_T \ dQ_T \ d\Pb_T \prod_{t=1}^{T-1} d\overline{Q}_t  \ d\Pb_t \\
&= \int m_T (d\overline{Q}_T - dQ_T) \ d\Pb_T \prod_{t=1}^{T-1} d\overline{Q}_t  \ d\Pb_t + \int m_{T-1} \prod_{t=1}^{T-1} d\overline{Q}_t  \ d\Pb_t \\
&= \sum_{t=1}^T \int m_t (d\overline{Q}_t - dQ_t) \ d\Pb_t \prod_{s=0}^{t-1} d\overline{Q}_t  \ d\Pb_t
\end{align*}
where the first equality follows by definition, the second by adding and subtracting the last term, the third by definition of $m_t$, and the fourth by repeating this process $T$ times.

Now we have that the expected contribution to the influence function due to estimating $Q$ when it depends on $\Pb$ is
\begin{align*}
\E \sum_{t=1}^T &\left( \prod_{s=0}^{t-1}  \frac{ d\overline{Q}_s }{ d\overline\pi_s } \right) \int \overline\phi_t \ \overline{m}_t \ d\nu = \sum_{t=1}^T \int \overline\phi_t \ d\pi_t \ \overline{m}_t \ d\nu \ d\Pb_t \left( \prod_{s=0}^{t-1}  \frac{ d\overline{Q}_s }{ d\overline\pi_s } d\pi_s \ d\Pb_s \right) \\
&= \sum_{t=1}^T \int (\overline\phi_t \ d\pi_t) (\overline{m}_t - m_t) \ d\nu \ d\Pb_t \left( \prod_{s=0}^{t-1}  \frac{ d\overline{Q}_s }{ d\overline\pi_s } d\pi_s \ d\Pb_s \right) \\
& \hspace{.4in} + \sum_{t=1}^T \int (\overline\phi_t \ d\pi_t) \ m_t \ d\nu \ d\Pb_t \left( \prod_{s=0}^{t-1}  \frac{ d\overline{Q}_s }{ d\overline\pi_s } d\pi_s \ d\Pb_s \right) \\
%&= \sum_{t=1}^T \int (\overline\phi_t \ d\pi_t) (\overline{m}_t - m_t) \ d\nu \ d\Pb_t \left( \prod_{s=0}^{t-1}  \frac{ d\overline{Q}_s }{ d\overline\pi_s } d\pi_s \ d\Pb_s \right) \\
%& \hspace{.4in} + \sum_{t=1}^T \int (\overline\phi_t \ d\pi_t) \ m_t \ d\nu \ d\Pb_t \left( \frac{d\pi_{t-1}-d\overline\pi_{t-1}}{d\overline\pi_{t-1}} \right) d\overline{Q}_{t-1} \ d\Pb_{t-1}  \left( \prod_{s=0}^{t-2}  \frac{ d\overline{Q}_s }{ d\overline\pi_s } d\pi_s \ d\Pb_s \right) \\
%& \hspace{.4in} + \sum_{t=1}^T \int (\overline\phi_t \ d\pi_t) \ m_t \ d\nu \ d\Pb_t \ d\overline{Q}_{t-1} \ d\Pb_{t-1}  \left( \prod_{s=0}^{t-2}  \frac{ d\overline{Q}_s }{ d\overline\pi_s } d\pi_s \ d\Pb_s \right) \\
&= \sum_{t=1}^T \int (\overline\phi_t \ d\pi_t) (\overline{m}_t - m_t) \ d\nu \ d\Pb_t \left( \prod_{s=0}^{t-1}  \frac{ d\overline{Q}_s }{ d\overline\pi_s } d\pi_s \ d\Pb_s \right) \\
& \hspace{.4in} + \sum_{t=1}^T \sum_{s=1}^t \int (\overline\phi_t \ d\pi_t) \ m_t \ d\nu \ d\Pb_t \left( \prod_{r=0}^{t-1} d\overline{Q}_r \ d\Pb_r \right) \left( \frac{d\pi_s-d\overline\pi_s}{d\overline\pi_s} \right)  \left( \prod_{r=0}^{s-1}  \frac{ d\pi_s }{ d\overline\pi_s } \right) \\
&\hspace{.4in} +  \sum_{t=1}^T \int (\overline\phi_t \ d\pi_t) \ m_t \ d\nu \ d\Pb_t \left( \prod_{s=0}^{t-1} dQ_s  \ d\Pb_s \right)
\end{align*}
where the first equality follows by iterated expectation, the second by adding and subtracting the second term in the sum, and the third %and fourth equalities 
by the same logic as in Lemma \ref{lem:rem_qknown}.

Now considering the last term in the above display plus $\psi(\Pb;\overline{Q}) - \psi(\Pb;Q)$ we have
\begin{align*}
\psi(\Pb;\overline{Q}) - \psi(\Pb;Q) &+ \sum_{t=1}^T \int (\overline\phi_t \ d\pi_t) \ m_t \ d\nu \ d\Pb_t \left( \prod_{s=0}^{t-1} dQ_s  \ d\Pb_s \right) \\
&= \sum_{t=1}^T \int  m_t (d\overline{Q}_t - dQ_t + \overline\phi_t \ d\pi_t) \ d\nu \ d\Pb_t \left( \prod_{s=0}^{t-1} dQ_s  \ d\Pb_s \right)
\end{align*}
which yields the result.
\end{proof}

Now we need to translate the remainder terms from Lemmas \ref{lem:rem_qknown} and \ref{lem:rem_contribution} to the incremental propensity score intervention setting. The remainder from Lemma \ref{lem:rem_qknown} equals
\begin{align*}
&\sum_{t=1}^T \sum_{s =1}^t \int \left( m_t^* - \overline{m}_t \right)\left( \frac{d\pi_s-d\overline\pi_s}{d\overline\pi_s} \right) \left( \prod_{r=1}^{s-1}   \frac{d\pi_r}{d\overline\pi_r} \right)  \left( \prod_{r=1}^t d\overline{Q}_r \ d\Pb_r \right) \\
&= \sum_{t=1}^T \sum_{s =1}^t \int \Big\{ (\overline{m}_{t+1} - {m}_{t+1}) d\overline{Q}_{t+1} d\Pb_{t+1} + m_{t+1}( d\overline{Q}_{t+1} - dQ_{t+1}) d\Pb_{t+1}  \\
&\hspace{.4in} +  (m_t - \overline{m}_t) \Big\}  \left( \frac{d\pi_s - d\overline\pi_s}{d\overline\pi_s} \right)  \left( \prod_{r=1}^{s-1}   \frac{d\pi_r}{d\overline\pi_r} \right)  \left[ \prod_{r=1}^t \left( \frac{d\overline{Q}_r}{d\pi_r} \right) d\pi_r \ d\Pb_r \right] \\
%&= \sum_{t=1}^T \sum_{s =1}^t \E\left\{ \left( \frac{ \overline{m}_{t+1} - m_{t+1} }{\pi_{t+1}/ \overline{q}_{t+1} } + m_{t+1} \frac{\overline{q}_{t+1}-q_{t+1}}{\pi_{t+1}} + m_t - \overline{m}_t \right) \left( \frac{\pi_s - \overline\pi_s}{\overline\pi_s} \right)  \left( \prod_{r=1}^{s-1}   \frac{\pi_r}{\overline\pi_r} \right)  \left( \prod_{r=1}^t  \frac{\overline{q}_r}{\pi_r} \right) \right\} \\
&\lesssim \sum_{t=1}^T \sum_{s =1}^t  \Big( \| \overline{m}_{t+1} - m_{t+1} \| + \|\overline\pi_{t+1} - \pi_{t+1} \| + \| m_t - \overline{m}_t \| \Big) \| \pi_s - \overline\pi_s \|
\end{align*}
where the last inequality follows since
$$ d\overline{Q}_t - dQ_t =  \frac{\delta(2 a_t - 1)(\overline\pi_t-\pi_t)}{ (\delta \overline\pi_t + 1-\pi_t)(\delta \pi_t + 1-\pi_t)} . $$

For the remainder from Lemma \ref{lem:rem_contribution} first note that
$$ \int \overline\phi_t \ d\pi_t = \frac{\delta (2 a_t-1) (\pi_t-\overline\pi_t)}{(\delta \overline\pi_t + 1-\overline\pi_t)^2}  $$
where we used the form of the efficient influence function derived in Lemma \ref{lem:eif_contribution}. Combining the two previous expressions gives
$$ d\overline{Q}_t - dQ_t + \int \overline\phi_t \ d\pi_t = \frac{\delta(\delta-1)(2 a_t - 1) (\overline\pi_t-\pi_t)^2}{ (\delta \overline\pi_t + 1-\overline\pi_t)^2 (\delta \pi_t + 1-\pi_t)} . $$
Thus the remainder from Lemma \ref{lem:rem_contribution} is
\begin{align*}
&\sum_{t=1}^T \int (\overline\phi_t \ d\pi_t) (\overline{m}_t - m_t) \ d\nu \ d\Pb_t \left( \prod_{s=0}^{t-1}  \frac{ d\overline{Q}_s }{ d\overline\pi_s } d\pi_s \ d\Pb_s \right) \\
& \hspace{.4in} + \sum_{t=1}^T \sum_{s=1}^t \int (\overline\phi_t \ d\pi_t) \left( \frac{d\pi_s-d\overline\pi_s}{d\overline\pi_s} \right)  m_t \ d\nu \ d\Pb_t \left( \prod_{r=0}^{t-1} d\overline{Q}_r \ d\Pb_r \right) \left( \prod_{r=0}^{s-1}  \frac{ d\pi_s }{ d\overline\pi_s } \right) \\
& \hspace{.4in} + \sum_{t=1}^T \int  m_t (d\overline{Q}_t - dQ_t + \overline\phi_t \ d\pi_t) \ d\nu \ d\Pb_t \left( \prod_{s=0}^{t-1} dQ_s  \ d\Pb_s \right) \\
& \lesssim \sum_{t=1}^T  \| \pi_t - \overline\pi_t \| \left( \| \overline{m}_t - m_t \|  + \sum_{s=1}^t \| \pi_s - \overline\pi_s \| + \| \pi_t - \overline\pi_t \| \right) .
\end{align*}

The condition given in Theorem \ref{thm:unif_clt}, that for $s \leq t \leq T$ we have
$$ \left( \sup_{\delta \in \mathcal{D}} \| \hat{m}_{t,\delta} - m_{t,\delta} \| + \| \hat\pi_t - \pi_t \| \right) \| \hat\pi_s - \pi_s \| = o_\Pb(1/\sqrt{n}) , $$
therefore ensures that the above remainders from Lemmas \ref{lem:rem_qknown} and \ref{lem:rem_contribution} are negligible up to order $n^{-1/2}$ uniformly in $\delta$. Therefore $\| B_{n,2} \|_\mathcal{D}=o_\Pb(1)$, which concludes the proof.

\subsection{Proof of Theorem \ref{thm:unif_ci}}
\label{proof:unif_ci}

As in the proof of Theorem \ref{thm:unif_clt}, let $\| f \|_\mathcal{D} = \sup_{\delta \in \mathcal{D}} | f(\delta) |$ denote the supremum norm with respect to $\delta$, and define the processes
\begin{align*}
\widehat\Psi_n(\delta) &= \sqrt{n} \{\hat\psi(\delta) - \psi(\delta)\} / \hat\sigma(\delta) \\
\widehat\Psi_n^*(\delta) &= \Gn[\xi \{ \varphi(\bZ;\hat{\boldsymbol\eta}_{\text{-}S},\delta)-\hat\psi(\delta)\} /\hat\sigma(\delta) ]  \\
\Psi_n^*(\delta) &= \Gn[\xi \{ \varphi(\bZ;\boldsymbol{\eta}, \delta)-\psi(\delta)\} /\sigma(\delta) ]  .
\end{align*}
Note that the star superscripts denote multiplier bootstrap processes. As before, let $\Gb(\delta)$ denote the mean-zero Gaussian process with covariance $\E\{ \widetilde\varphi(\bZ;\boldsymbol\eta,\delta_1)\widetilde\varphi(\bZ;\boldsymbol\eta,\delta_2) \}$. 

Since
\begin{align*}
&\Pb \left\{ \hat\psi(\delta) - \frac{ \hat{c}_\alpha \hat\sigma(\delta)}{\sqrt{n} } \leq \psi(\delta) \leq \hat\psi(\delta) +  \frac{ \hat{c}_\alpha \hat\sigma(\delta)}{\sqrt{n} }, \text{ for all } \delta \in \mathcal{D} \right\} \\
& \hspace{.4in} = \Pb\left( \sup_{\delta \in \mathcal{D}} \left| \frac{\hat\psi(\delta) - \psi(\delta)}{\hat\sigma(\delta)/\sqrt{n}} \right| \leq \hat{c}_\alpha \right) = \Pb\left(  \| \widehat\Psi_n \|_\mathcal{D} \leq \hat{c}_\alpha \right) ,
\end{align*}
the result of Theorem \ref{thm:unif_ci} requires that we show
$$ \left| \Pb\left(  \| \widehat\Psi_n \|_\mathcal{D} \leq \hat{c}_\alpha \right) - \Pb\left(  \| \widehat\Psi^*_n \|_\mathcal{D} \leq \hat{c}_\alpha \right) \right| =  o(1) , $$
which yields the desired result since $\Pb(  \| \widehat\Psi^*_n \|_\mathcal{D} \leq \hat{c}_\alpha )=1-\alpha$ by definition of $\hat{c}_\alpha$. 

We showed in the proof of Theorem \ref{thm:unif_clt} that $ \| \widehat\Psi_n  - \Psi_n \|_\mathcal{D} = o_\Pb(1) $ which implies that
 $| \| \widehat\Psi_n \|_\mathcal{D} - \| \Psi_n \|_\mathcal{D} | = o_\Pb(1)$, 
and by Corollary 2.2 in \textcite{chernozhukov2014gaussian} we have
$$ \Big| \| \Psi_n \|_\mathcal{D} - \| \Gb \|_\mathcal{D} \Big| = o_\Pb(1) . $$
Hence by Lemma 2.3 in  \textcite{chernozhukov2014gaussian} it follows that
$$ \sup_{t \in \R} \left| \Pb\Big( \| \widehat\Psi_n \|_\mathcal{D} \leq t \Big) - \Pb\Big( \| \Gb \|_\mathcal{D} \leq t \Big) \right| = o(1) . $$
Similarly, by Corollary 2.2 of \textcite{belloni2015uniformly} we have
$$ \Big| \| \widehat\Psi^*_n \|_\mathcal{D}  - \| \Psi^*_n \|_\mathcal{D} \Big| = o_\Pb(1) \ , \  \Big| \| \Psi^*_n \|_\mathcal{D} - \| \Gb \|_\mathcal{D} \Big| = o_\Pb(1) $$
so that again by Lemma 2.3 in  \textcite{chernozhukov2014gaussian}
$$ \sup_{t \in \R} \left| \Pb\Big( \| \widehat\Psi^*_n \|_\mathcal{D} \leq t \Big) - \Pb\Big( \| \Gb \|_\mathcal{D} \leq t \Big) \right| = o(1). $$
This yields the result, since $| \Pb(  \| \widehat\Psi_n \|_\mathcal{D} \leq \hat{c}_\alpha ) - \Pb(  \| \widehat\Psi^*_n \|_\mathcal{D} \leq \hat{c}_\alpha ) | $ is bounded above by
$$ \sup_{t \in \R} \left| \Pb\Big( \| \widehat\Psi_n \|_\mathcal{D} \leq t \Big) - \Pb\Big( \| \Gb \|_\mathcal{D} \leq t \Big) \right| +  \sup_{t \in \R} \left| \Pb\Big( \| \widehat\Psi^*_n \|_\mathcal{D} \leq t \Big) - \Pb\Big( \| \Gb \|_\mathcal{D} \leq t \Big) \right| = o(1) . $$

\pagebreak
\subsection{R Code}
\label{sec:code}

\bigskip

\begin{verbatim}
### this function requires the following inputs:
###   dat: dataframe (in long not wide form if longitudinal) with columns
###     `time', `id', outcome `y', treatment `a'
###   x.trt: covariate matrix for treatment regression
###   x.out: covariate matrix for outcome regression
###   delta.seq: sequence of delta values
###   nsplits: number of sample splits
### NOTE: dat, x.trt, x.out should all have the same number of rows

ipsi <- function(dat, x.trt, x.out, delta.seq, nsplits){

# setup storage
ntimes <- length(table(dat$time)); n <- length(unique(dat$id))
k <- length(delta.seq); ifvals <- matrix(nrow=n,ncol=k); est.eff <- rep(NA,k)
wt <- matrix(nrow=n*ntimes,ncol=k); cumwt <- matrix(nrow=n*ntimes,ncol=k)
rt <- matrix(nrow=n*ntimes,ncol=k); vt <- matrix(nrow=n*ntimes,ncol=k)

s <- sample(rep(1:nsplits,ceiling(n/nsplits))[1:n])
slong <- rep(s,rep(ntimes,n))

for (split in 1:nsplits){ print(paste("split",split)); flush.console()

# fit treatment model
trtmod <- ranger(a ~ ., dat=cbind(x.trt,a=dat$a)[slong!=split,])
dat$ps <- predict(trtmod, data=x.trt)$predictions 

for (j in 1:k){ print(paste("delta",j)); flush.console()
delta <- delta.seq[j]

# compute weights
wt[,j] <- (delta*dat$a + 1-dat$a)/(delta*dat$ps + 1-dat$ps)
cumwt[,j] <- as.numeric(t(aggregate(wt[,j],by=list(dat$id),cumprod)[,-1]))
vt[,j] <- (1-delta)*(dat$a*(1-dat$ps) - (1-dat$a)*delta*dat$ps)/delta

# fit outcome models
outmod <- vector("list",ntimes); rtp1 <- dat$y[dat$time==end]
print("fitting regressions"); flush.console()
for (i in 1:ntimes){ 
  t <- rev(unique(dat$time))[i]
  outmod[[i]] <- ranger(rtp1 ~ ., 
    dat=cbind(x.out,rtp1)[dat$time==t & slong!=split,])
  newx1 <- x.out[dat$time==t,]; newx1$a <- 1
  m1 <- predict(outmod[[i]], data=newx1)$predictions
  newx0 <- x.out[dat$time==t,]; newx0$a <- 0
  m0 <- predict(outmod[[i]], data=newx0)$predictions
  pi.t <- dat$ps[dat$time==t]
  rtp1 <- (delta*pi.t*m1 + (1-pi.t)*m0) / (delta*pi.t + 1-pi.t)
  rt[dat$time==t,j] <- rtp1 }
  
ifvals[s==split,j] <- ((cumwt[,j]*dat$y)[dat$time==end] + 
  aggregate(cumwt[,j]*vt[,j]*rt[,j],by=list(dat$id),sum)[,-1])[s==split]

} }

# compute estimator
for (j in 1:k){ est.eff[j] <- mean(ifvals[,j]) }

# compute asymptotic variance
sigma <- sqrt(apply(ifvals,2,var))
eff.ll <- est.eff-1.96*sigma/sqrt(n); eff.ul <- est.eff+1.96*sigma/sqrt(n)

# multiplier bootstrap
eff.mat <- matrix(rep(est.eff,n),nrow=n,byrow=T)
sig.mat <- matrix(rep(sigma,n),nrow=n,byrow=T)
ifvals2 <- (ifvals-eff.mat)/sig.mat
nbs <- 10000; mult <- matrix(2*rbinom(n*nbs,1,.5)-1,nrow=n,ncol=nbs)
maxvals <- sapply(1:nbs, function(col){ 
  max(abs(apply(mult[,col]*ifvals2,2,sum)/sqrt(n))) } )
calpha <- quantile(maxvals, 0.95)
eff.ll2 <- est.eff-calpha*sigma/sqrt(n); eff.ul2 <- est.eff+calpha*sigma/sqrt(n)

return(list(est=est.eff, sigma=sigma, ll1=eff.ll,ul1=eff.ul, 
  calpha=calpha, ll2=eff.ll2,ul2=eff.ul2))

}
\end{verbatim}

\end{document}